\documentclass[a4paper,twocolumn,11pt,accepted=2026-09-03]{quantumarticle}
\pdfoutput=1
\usepackage[utf8]{inputenc}
\usepackage[english]{babel}
\usepackage[T1]{fontenc}
\usepackage{authblk}
\usepackage{amsmath}
\usepackage{amssymb}
\usepackage{revsymb}
\usepackage{graphicx}
\usepackage{amsthm}
\usepackage{changepage}
\usepackage{enumerate}
\usepackage{color}
\usepackage{physics}
\usepackage{parcolumns}
\usepackage{mathrsfs}
\usepackage{multirow}
\usepackage[colorlinks,linktocpage,hypertexnames=true,pdfpagelabels]{hyperref}
\usepackage{subcaption}
\usepackage[capitalise]{cleveref}
\usepackage{tikz-cd}
\usepackage{environ}
\usepackage{tabularx}
\usepackage{complexity}
\usepackage{comment}

\theoremstyle{plain}
\newtheorem{thm}{Theorem}[section]
\newtheorem{prop}[thm]{Proposition}
\newtheorem{lemma}[thm]{Lemma}
\newtheorem{coro}[thm]{Corollary}
\newtheorem{defi}[thm]{Definition}
\newtheorem{example}[thm]{Example}

\AddToHook{env/lemma/begin}{\crefalias{thm}{lemma}}
\AddToHook{env/remark/begin}{\crefalias{thm}{remark}}
\AddToHook{env/prop/begin}{\crefalias{thm}{prop}}
\AddToHook{env/coro/begin}{\crefalias{thm}{coro}}
\AddToHook{env/defi/begin}{\crefalias{thm}{defi}}
\AddToHook{env/example/begin}{\crefalias{thm}{example}}

\crefname{thm}{Theorem}{Theorems}
\crefname{lemma}{Lemma}{Lemmas}
\crefname{prop}{Proposition}{Propositions}
\crefname{coro}{Corollary}{Corollaries}
\crefname{defi}{Definition}{Definitions}
\crefname{example}{Example}{Examples}

\numberwithin{equation}{section}

\theoremstyle{definition}
\newtheorem{remark}[thm]{Remark}
\crefname{remark}{Remark}{Remarks}

\newcommand{\states}[1]{\S(#1)}

\renewcommand{\C}{\mathbb{C}}
\renewcommand{\R}{\mathbb{R}}

\newcommand{\I}{\openone}

\newcommand{\mc}{\mathcal}

\renewcommand{\H}{\mathcal{H}}
\newcommand{\Kr}{E}

\newcommand{\B}{\mathcal{B}}

\renewcommand{\S}{\mc{S}}

\newcommand{\cE}{\mathcal{E}}
\newcommand{\fmig}{\frac{1}{2}}

\newcommand{\swich}[1]{\left(#1\right)}
\newcommand{\cwich}[1]{\left[#1\right]}
\newcommand{\set}[1]{\left\{#1\right\}}
\newcommand{\sprod}[2]{\langle #1,#2\rangle}

\newcommand{\id}{\operatorname{id}}

\newcommand{\jamiol}{{Jamio{\l}kowski}}
\newcommand{\cone}[1]{\mathrm{cone}\swich{#1}}
\newcommand{\supp}{\operatorname{supp}}

\definecolor{redQ}{RGB}{231, 43, 72}
\newcommand{\jcc}[1]{\textcolor{redQ}{#1}}
\definecolor{amethyst}{rgb}{0.75, 0, 1}
\newcommand{\mh}[1]{\textcolor{amethyst}{#1}}

\begin{document}
\title{Approach to optimal quantum transport via states over time}

\author[1]{{Matt Hoogsteder-Riera}}\email{hoogstedermatt@gmail.com}
\author[1]{{John Calsamiglia}}\email{john.calsamiglia@uab.cat}
\author[1,2,3,4]{{Andreas Winter}}\email{andreas.winter@uab.cat}
\affil[1]{{\textit{Grup d'Informaci\'o Qu\`antica, Departament de F\'isica,\protect\\[-1mm] Universitat Aut\`onoma de Barcelona, 08193 Bellaterra (BCN), Spain}}}
\affil[2]{{\textit{ICREA---Instituci\'o Catalana de Recerca i Estudis Avan{\c{c}}ats,\protect\\[-1mm] Pg.~Llu\'is Companys, 23, 08010 Barcelona, Spain}}}
\affil[3]{{\textit{Department Mathematik/Informatik---Abteilung Informatik,\protect\\[-1mm] Universit\"at zu K\"oln, Albertus-Magnus-Platz, 50923 K\"oln, Germany}}}
\affil[4]{{\textit{Institute for Advanced Study, Technische Universit\"at M\"unchen,\protect\\[-1mm] Lichtenbergstra{\ss}e 2a, 85748 Garching, Germany}}}

\date{(04 April 2025)}

\maketitle
\onecolumn

\begin{abstract}
We approach the problem of constructing a quantum analogue of the 
immensely fruitful classical transport cost theory of Monge from a 
new angle. Going back to the original motivations, by which the transport 
is a bilinear function of a mass distribution (without loss of generality 
a probability density) and a transport plan (a stochastic kernel), 
we explore the quantum version where the mass distribution is generalised 
to a density matrix, and the transport plan to a completely positive 
and trace preserving map. 
These two data are naturally integrated into their Jordan product, 
which is called \emph{state over time} (``stote''), 
and the transport cost is postulated to be a linear function of it. 
We explore the properties of this transport cost, as well as the optimal 
transport cost between two given states (simply the minimum cost over all 
suitable transport plans). After that, we analyse in considerable detail 
the case of unitary invariant cost, for which we can calculate many 
costs analytically. These findings suggest that our quantum transport 
cost is qualitatively different from Monge's classical transport. 
\end{abstract}

\section{Introduction}
\label{sec:intro}
The Monge problem, and the field of optimal transport it has spawned \cite{monge81,kantorovich48,kantorovich58,villani08}, asks for the best way to move mass, modeled as probability distributions, according to some cost functional. Explicitly, the optimal transport cost between two probability distributions $\mu,\nu$ defined over space $X$ is 
\begin{equation}
  \label{eq:classical-cost}
  c(\mu,\nu) := \inf_{\pi\in C(\mu,\nu)} \int_{X\times X}k(x,y)d\pi(x,y),
\end{equation} 
where $k(x,y)$ is a (usually real and positive) cost function and $C(\mu,\nu)$ is the set of 
\emph{couplings} between $\mu$ and $\nu$, that is the probability distributions on $X\times X$ 
with marginals $\mu$ and $\nu$. 

This approach allows for properties of the space $X$ to be reflected in the cost through the function $k:X\times X\rightarrow \R$. In contrast, typical distinguishability measures on probability spaces, such as the total variation distance,\footnote{Note that the total variation distance can be defined in terms of optimal transport cost, by choosing the cost to be the trivial metric, $k(x,y)=1-\delta_{xy}$.} or the Kullback-Leibler divergence, can be computed focusing solely on the probability distributions and ignoring the underlying properties of the space $X$. As an example, let $X=\R$ and $\mu_0=\delta_0$, $\mu_1=\delta_1$ and $\nu=\delta_{-1}$. The total variation distances are equal: $\Delta_{TV}(\mu_0,\nu)=\Delta_{TV}(\mu_1,\nu)=1$, but we can define an optimal transport cost that reflects the Euclidean metric of the reals by choosing $k(x,y)=\abs{x-y}$. With this cost function, $\Delta_{OT}(\mu_0,\nu)= 1$, but $\Delta_{OT}(\mu_1,\nu)= 2$. This difference of the distances between $0$ and $-1$ and between $1$ and $-1$ is ignored by the total variation distance, but taken into account in the optimal transport cost.

Physically, the cost function $k$ could be reflecting things like energy or information related quantities such as the number of bit flips necessary to transform a bit string into another one (known as the Hamming distance), the latter being a special case of a metric on $X$. 

In the present paper, we make (yet another) attempt to quantise the theory of optimal classical transport from probability densities to quantum states. 
Our interpretation of the coupling $\pi(x,y)$ appearing in Eq.~\eqref{eq:classical-cost} is that it defines a stochastic matrix that yields output $\nu$ for input $\mu$. We can recover this transformation explicitly using Bayes' Theorem: $p(y|x)=\pi(x,y)/\mu(x)$, and this is the ``transport plan'': the map describing which fraction of mass at each given point is transferred to a given target point. The cost appearing on the right hand side of Eq.~\eqref{eq:classical-cost} is then bilinear in $\mu$ and $p(y|x)$, and we attempt to preserve this feature in our quantum version.

There have been attempts to generalise optimal transport to quantum systems dating back to the late 1990s \cite{zyczkowski98,zyczkowski01}. Some more recent approaches look at the problem directly through the primal formulation \cite{hoogsteder18,depalma21,friedland22,bistron22,bunth25} (some form of couplings), through the dual formulation \cite{depalma21b,depalma23} (which \eqref{eq:classical-cost} has via linear programming duality) and through the continuous formulation \cite{benamou00,carlen12,rouze17,caglioti21} (of certain dynamical semigroups having trajectories that are geodesic for a suitably defined optimal transport distance). Our work looks at the problem through the primal formulation. Our objective here is to define a quantum optimal transport where couplings have a straightforward physical interpretation as bilinear combinations of quantum states and quantum channels, similar to our interpretation of the classical couplings, leading to a cost assigned to each transport plan that is bilinear in the initial density and the quantum channel. 

In \cref{sec:motivation} we motivate the mathematical and conceptual idea behind our formulation, in \cref{sec:definition} we make the basic definitions, in \cref{sec:properties} we show our main results regarding properties of our proposed quantum optimal transport, in \cref{sec:UIcost} we discuss a specific cost metric that results from imposing unitary invariance and in \cref{sec:conclusions} we discuss the messages to take away and open problems derived from our work.

\subsection{Mathematical preliminaries}
\label{sec:preliminaries}

The Choi and {\jamiol} isomorphisms \cite{jamiolkowski72,choi75} are a relation between CPTP maps from $\B(\H_A)$ to $\B(\H_B)$ and bipartite states on the joint state space $\S(\H_A\otimes\H_B)$. These isomorphisms are relevant to this work because they allow us to treat maps as operators, allowing for operations that would not be possible otherwise. In more general terms, these isomorphisms allow us to translate complete positivity of a map to positivity of an operator, with the latter usually being much easier to check. We also make use of this fact in \cref{sec:definition} and \cref{sec:UIcost}.

\begin{thm}[Choi isomorphism \cite{choi75}]
\label{thm:ChoiIso}
    Let $\cE:\B(\H_A)\rightarrow\B(\H_B)$ be a linear map and $\ket{\Phi_+}\in \H_A\otimes\H_{A'}$ the non-normalised maximally entangled vector on the composite system of $A$ and a copy $A'\simeq A$ of it. The Choi operator, defined as 
    \begin{equation}
    \label{eq:ChoiMatrix}
        C_{\cE} = \swich{\id_A\otimes\cE}\swich{\ketbra{\Phi_+}} \in \B(\H_A\otimes\H_B),
    \end{equation} 
    recovers the action of $\cE$ as \begin{equation}
        \cE(x)=\Tr_A\cwich{(x^T\otimes\I_B)C_{\cE}}.
    \end{equation} 
    Moreover, $\cE$ is completely positive if and only if $C_{\cE}$ is positive semi-definite and $\cE$ is trace preserving if and only if $\;\Tr_B\cwich{C_{\cE}}=\I_A$.
\end{thm} 
{\jamiol}'s statement of the same result is just the Choi isomorphism with a partial transpose. 
\begin{thm}[{\jamiol} isomorphism \cite{jamiolkowski72}]
\label{thm:JamiolIso}
    Let $\cE:\B(\H_A)\rightarrow\B(\H_B)$ be a linear map and $\S\in\B(\H_A\otimes \H_{A'})$ the swap operator on the joint system $A\otimes A'$. The {\jamiol} operator,  defined as 
    \begin{equation}
      \label{eq:def-jamiol}  
      J_{\cE} = \swich{\id_A\otimes\cE}\swich{\S} \in \B(\H_A\otimes\H_B),
    \end{equation} 
    recovers the action of $\cE$ as \begin{equation}
        \cE(x)=\Tr_A\cwich{(x\otimes\I_B)J_{\cE}}.
    \end{equation} 
    Moreover, $\cE$ is completely positive if and only if $\;J_{\cE}^{T_A}$ is positive\footnote{For simplicity we usually use positive instead of the more rigorous positive semi-definite.} and $\cE$ is trace preserving if and only if $\Tr_B\cwich{J_{\cE}}=\I_A$.
\end{thm}

It is clear from the theorems that if $C$ and $J$ are Choi and {\jamiol} operators associated to the same CPTP map, then \begin{equation}
    J=C^{T_A}.
\end{equation} 

These theorems clarify how states and channels relate and how the two isomorphisms relate to each other. Each has its own advantages: the Choi matrix is manifestly positive, which is convenient when positivity must be checked explicitly, while the {\jamiol} form is basis-independent. In the Choi construction  one uses a maximally entangled state and the transpose on $\mathcal{H}_A$. Both of those steps depend on a choice of basis, so the Choi matrix is basis-dependent. In contrast, using the {\jamiol} form avoids that dependence. Moreover, the swap operator plays a useful computational role in the {\jamiol} picture, facilitating manipulations like exchanging tensor factors or expressing compositions more neatly.

We use the notation $C$ and $J$ for Choi and {\jamiol} operators, respectively, throughout this work. Moreover, we generally refer to Choi or {\jamiol} operators associated to CPTP map. Because our work deals with finite dimensional systems only, we also use the terms Choi and {\jamiol} matrix interchangeably with operator. We usually use the {\jamiol} operator, but the Choi operator will sometimes be more convenient to use in some circumstances. We always assume that the transpose is taken with respect to the canonical basis, or whichever privileged basis the problem presents. The basis dependence of the Choi matrix is only a problem when dealing with practical cases and never for general theoretical derivations. 

Throughout the paper we use $A^T$ for the transpose of a matrix $A$, $\overline{A}$ for the element-wise complex conjugate, and $A^*$ for the conjugate transpose. With this notation, $A^*=\overline{A}^T$.

\section{Jordan product motivation and properties}
\label{sec:motivation}

In this section we want to motivate our choice of coupling for the quantum optimal transport, since this choice is what differentiates our approach from the rest. The main observation was done in the previous section: a classical joint probability distribution can be interpreted as both a correlation function of a composite system and a map for the evolution over time of a single subsystem. The first interpretation can be generalised to quantum systems with bipartite states, as seen in, for example, \cite{friedland22,bunth25}. We choose the other interpretation, that of a stochastic map on a system. In the context of quantum mechanics this corresponds to quantum channels acting on a system. As we explain in the following, these questions on how to write states encoding the correlations of an input system with its output in quantum mechanics had been asked before in the context of quantum foundations.

The concept of \emph{states over time} as introduced in \cite{horsman17} is based on previous 
attempts to formalise causal correlations in quantum theory, see \cite{leifer13,fitzsimons15} 
and references within \cite{horsman17}. Its main motivation is to find an operator in the tensor 
product of the state spaces associated to two time points, which would capture the 
process transforming the initial state to the final one.

For this purpose, several properties have been put forward as desirables for a state over time: as a function of the initial state and the quantum channel it should be Hermitian preserving, bilinear in the two arguments, it should contain the classical case, reproduce the initial and final states as marginals, and, finally, keep the composability of channels. Through the later works of Fullwood and Parzygnat \cite{fullwood22,fullwood23,fullwood23b} and \cite{lie23}, we know that the Jordan product (also called symmetric bloom or Fullwood-Parzygnat state over time function in the literature) is the unique function that fulfils a slightly stronger set of axioms \cite{lie23} that imply the properties above (that had been proposed in \cite{horsman17}). 

The set of desiderata introduced in the paper are, in our opinion, well suited to the study of quantum optimal transport. Notably, the  bilinearity, associativity (which translates to clear composability of the channels) and clear reduction to the classical case. Hermiticity and preservation of marginal states are useful mathematically as they allow us to define the problem as an SDP. Other approaches mentioned in \cite{horsman17}, the LS state over time and the W state over time fail to fulfil bilinearity and composability; and preservation of the classical limit, respectively. That said, other works \cite{depalma21,friedland22} have approached this problem using the LS state over time, even if not explicitly acknowledged in the manuscripts, with some success. Our approach takes the  generalised FJV state over time (from \cite{horsman17}) due to its nice properties not shared with other states over time.

In the context of quantum optimal transport, a state over time gives an initial state and a process (in the form of a CPTP map). As previously mentioned, this is similar to how in classical (optimal) transport a coupling $\pi(x,y)$ gives an initial state and final states (the marginals) and a process through the conditional probability formula $\pi(y|x)=\pi(x,y)/\pi(x)$. In this sense, a classical joint probability distribution acts as both a joint state in space and time, something that does not happen for quantum joint states. In this work we 
take the state over time interpretation of a classical joint probability distribution and 
extend this to the quantum case. 

The Jordan product is a defined as one half of the anti-commutator: given two operators $A,B$, we denote $A\star B := \fmig\{A,B\}= \fmig \swich{AB+BA}$ \cite{Jordan:nicht-assoz,Jordan:1,Jordan:2,JordanvonNeumannWigner}.
In the context of transport maps, this operation has several desirable properties, as previously 
mentioned here and further expanded upon in \cite{horsman17}.
Moreover, while the Jordan product is not associative in general, it was shown in \cite{fullwood22} 
to fulfil the associative property for products of matrices of the form 
$A_{01} \equiv A_{01}\otimes \I_{23}$, 
$B_{12} \equiv B_{12}\otimes \I_{03}$ and 
$C_{23} \equiv \I_{01}\otimes C_{23}$: 
$A_{01} \star (B_{12} \star C_{23}) = (A_{01} \star B_{12}) \star C_{23}$. 
This form of product is what we encounter in our formalism, see Subsection \ref{subsec:multipleTimes} for the details on the associativity. 

\subsection{States over two times}
\label{subsec:twoTimes}
Firstly, we formally define a state over time, for which we propose the handy 
abbreviation \emph{stote} (see Fig.~\ref{fig:stoat}): 

\begin{defi}
Let $\H_A,\H_B$ be finite dimensional Hilbert spaces and $\mc{J}(\H_A\rightarrow\H_B)$ the set of {\jamiol} matrices between these two spaces.
Let $\rho\in\states{\H_A}$ and $J\in\mc{J}(\H_A\rightarrow\H_B)$.
The associated \emph{state over time (stote)} is defined as 
$Q=\swich{\rho\otimes\I}\star J$. 
Typically we will omit the identity and write this as $Q=\rho\star J$. 
The set of states over time between two Hilbert spaces $\H_A,\,\H_B$ is the set of all operators of this form: 
\begin{align}
    \mc{Q}(\H_A:\H_B)=\left\{ \rho\star J | \,J\in\mc{J}(\H_A\rightarrow\H_B) \text{ and }  \rho\in S(\H_A)\right\}.
\end{align}
\end{defi} 

\begin{figure}[ht]
\centering
  \includegraphics[width=0.5\textwidth]{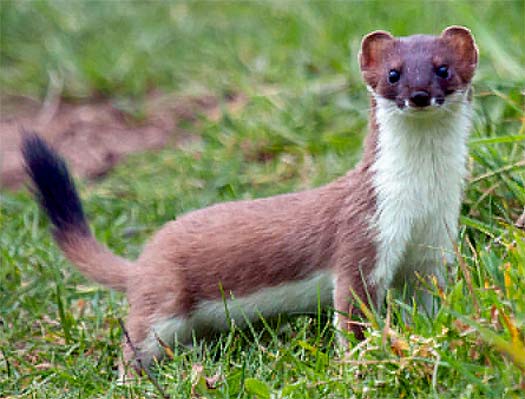}
  \caption{Stoat (also stote in old spelling), \emph{Mustela erminea};
           not to be confused with the common weasel, \emph{Mustela nivalis}, or other species of weasels, ferrets, minks or martens.} 
  \label{fig:stoat}
\end{figure}

This definition has the same interpretation as the classical probability coupling. 
Initial and final states appear as marginals: $\Tr_B\cwich{Q}=\rho$ and $\Tr_A \cwich{Q}=\sigma$, 
and the map connecting one to the other can be reconstructed from $Q$ as well
(see later in this section). 

In the literature of states over time, the {\jamiol} matrix has been used instead of the Choi matrix, and we will generally do the same. We denote Choi matrices by $C$ and {\jamiol} matrices by $J$, and they are related by $C=J^{T_A}$, as seen in \cref{sec:preliminaries}. 

Similar to the classical case, we are concerned with the set of channels that map a given state to another given state. We formalise this in the following definition.

\begin{defi}
Let $\rho, \sigma\in\S(\H_A),\S(\H_B)$ and $\mc{J}(\H_A\rightarrow\H_B)=\{J\in\B(\H_A\otimes\H_B),\,J^{T_A}\geq0, \, \Tr_B\cwich{J}=\I\}$ be the set of {\jamiol} matrices. The set of states over time between $\rho$ and $\sigma$ is \begin{align}
    \begin{split}
    \mc{Q}(\rho,\sigma)
    =&\set{\rho\star J | \,J\in\mc{J}(\H_A\rightarrow\H_B) \text{ and }  \Tr_A\cwich{\rho J}=\sigma}\\
    =&\left\{Q\in\mc{Q}(\H_A:\H_B)|\,\Tr_B\cwich{Q}=\rho, \,\Tr_A\cwich{Q}=\sigma\right\}.
    \end{split}
\end{align}  
\end{defi}

The last condition, $\Tr_A\cwich{Q}=\sigma$, specifies the image of $\rho$ under the associated quantum channel. That is, we are selecting the channels (in matrix form) that map $\rho$ to $\sigma$. We can also go in the opposite direction. That is, given a state over time, an associated initial state and {\jamiol} matrix (and therefore final state) can be found explicitly. This has been also shown in \cite{fullwood23,lie23}. 

\begin{thm}
\label{thm:JordanInverse}
Let $\omega$ be a Hermitian operator on $\B(\H_A\otimes\H_B)$ such that $\rho=\Tr_B\cwich{\omega}\geq0$. 
Then let $B=\{\ket{ik}\}$ be a product basis of $\H_A\otimes \H_B$ such that $\{\ket{i}\}$ is a diagonal basis of $\rho$ with associated eigenvalues $\{p_i \geq 0\}$. Finally, consider an operator\footnote{Note we are looking for a {\jamiol} operator here associated to a CPTP, but the theorem works regardless. The check is performed later, depending on whether $\rho$ is faithful or not.} $J$ such that, if $p_i$ or $p_j$ are nonzero then
\begin{align}
    \bra{ik}J\ket{j\ell}=\frac{2}{p_i+p_j}\bra{ik}\omega\ket{j\ell}.
\end{align}

Then, $\omega=\rho\star J$. Moreover, if $\rho$ is faithful then $J$ is Hermitian and the unique operator such that $\omega=\rho\star J$ for the given $\omega$.
\end{thm}

\begin{proof}
This is immediate from the construction of $\rho\star J$ in a product basis that contains an eigenbasis of $\rho$. 
Let $J$ be written in a product basis $\{\ket{ik}\}$ whose $\H_A$ component is a diagonal basis of $\rho$,
\begin{align}
    \rho&=\sum_i p_i\ketbra{i} \\
    J&=\sum_{ikj\ell}\bra{ik}J\ket{j\ell}\ketbra{ik}{j\ell}.
\end{align} 
In this basis we can calculate the Jordan product 
\begin{align}\begin{split}
    \rho\star J &= \fmig \swich{\sum_{ikj\ell i'}p_{i'}\ketbra{i'}\bra{ik}J\ket{j\ell}\ketbra{ik}{j\ell}+\bra{ik}J\ket{j\ell}\ketbra{ik}{j\ell}p_i\ketbra{i'}}\\
    &=\sum_{ikj\ell}\fmig(p_i+p_j)\bra{ik}J\ket{j\ell}\ketbra{ik}{j\ell}.
\end{split}\end{align} 
Therefore, if $p_i$ or $p_j$ are nonzero, the coefficients of $J$ from this matrix are 
\begin{align}
    \bra{ik}J\ket{j\ell}=\frac{2}{p_i+p_j}\bra{ik}\rho\star J\ket{j\ell}.
\end{align}

If $\rho$ is faithful, this fully characterises every coefficient of $J$ in the chosen basis, 
and therefore $J$ is unique.
\end{proof}

If $\rho$ is faithful, from \cref{thm:JordanInverse} we can check if $\rho$ and $J$ are a state and a {\jamiol} matrix, respectively, to conclude whether or not $\omega$ is a state over time. Note that the trace condition is guaranteed: \begin{equation}
\begin{split}
\Tr_B\cwich{J} &=\sum_{ikjl}\bra{ik}J\ket{j\ell}\ketbra{ik}{j\ell}= \sum_{ikjl}\frac{2}{p_i+p_j}\bra{ik}\omega\ket{j\ell}\ketbra{i}{j}\braket{\ell}{k}\\
&=\sum_{ij}\frac{2}{p_i+p_j}\bra{i}\swich{\sum_k\expval{\omega}{k}}\ket{j}\ketbra{i}{j} 
=\sum_{ij}\frac{2}{p_i+p_j}\bra{i}\Tr_B\cwich{\omega}\ket{j}\ketbra{i}{j} \\
&= \sum_{ij}\frac{2}{p_i+p_j}\bra{i}\rho\ket{j}\ketbra{i}{j} =\sum_{ij} \frac{2}{p_i+p_j}\delta_{ij}p_i\ketbra{i}{j}=\I,
\end{split}
\end{equation} where we have used that $\rho=\Tr\cwich{\omega}$ and its faithfulness.

In the case where $\rho$ is not faithful we end up with a matrix completion problem that can be solved with the following semidefinite programming (SDP) problem \cite{vandenberghe96,watrous11} 
\begin{align}\begin{split}
    \min_J &\quad \quad\; f(J) \\
    \text{s.t.} & \quad\:\left\{\begin{aligned}
        &\bra{ik}J\ket{j\ell}=\frac{2}{p_i+p_j}\bra{ik}\omega\ket{j\ell}, \quad \forall i,j \in B \;|\; p_i+p_j\neq0,\;\forall k,\ell\in B\\
        &\Tr_B J=\I\\
        &J^{T_A}\geq0
       \end{aligned} \right.
\end{split}\end{align} 
Here, $f$ is any linear function, since we are not interested in minimising a specific function but just in finding a matrix that fulfils the given conditions (a feasible solution). 
For numerical calculation purposes, we can rewrite this feasibility problem by adding an extra real variable $x$. 
This is useful because numerical solvers require the feasible set to have a nonempty interior. 
In some cases (like when $\rho$ is faithful) the set of feasible {\jamiol} matrices can have an empty interior and adding the dummy variable $x$ allows us to expand the feasible set. $x$ is added as follows:
\begin{align}\begin{split}
    \min_{(x,J)} &\quad \quad\; -x \\
    \text{s.t.} & \quad\:\left\{\begin{aligned}
        &\bra{ik}J\ket{j\ell}=\frac{2}{p_i+p_j}\bra{ik}\omega\ket{j\ell}, \quad \forall i,j \in B \;|\; p_i+p_j\neq0,\;\forall k,\ell\in B\\
        &\Tr_B J=\I\\
        &J^{T_A}\geq x\I
       \end{aligned} \right..
\end{split}\end{align}
From this it is clear that if the output of the SDP is a nonnegative $x$ then 
the associated matrix $J$ will be a {\jamiol} matrix.

In case that $\rho$ is faithful, the following basis independent expression from 
\cite{scandi23} can also be used: 
\begin{align}
    J = \int_{0}^\infty e^{-\frac{t}{2}\rho}\omega e^{-\frac{t}{2}\rho}dt.
\end{align}

As a corollary of this form we can see that this map is CP:
\begin{coro}\label{coro:JInverseCP}
    Let $\rho\in\S(\H)$ be a faithful state of a Hilbert space $\H$. Then the map 
    \begin{equation}
        x\mapsto \int_{0}^\infty e^{-\frac{t}{2}\rho} x e^{-\frac{t}{2}\rho}dt
    \end{equation} 
    is completely positive.
\end{coro}

\begin{proof}
    $e^{-\frac{t}{2}\rho}$ will be Hermitian because $\rho$ is. For a fixed $t$, let $K=K^\dagger=e^{-\frac{t}{2}\rho}$ the map associated to this fixed $t$, $x\mapsto Kx K^\dagger$, will be CP because it is written in Kraus form \cite[Ch.~8]{kraus83}, \cite{nielsen10}. The integral of CP maps will be CP, thus the original map is CP. 
\end{proof}

Still for a faithful state $\rho$, starting form the canonical basis, the following expression is also equivalent and will be useful later: 
\begin{align}
     J = \swich{U_\rho\otimes \I}\swich{U_\rho^*\rho U_\rho\star \swich{\swich{U_\rho^*\otimes\I} \omega \swich{U_\rho\otimes \I}}^\Theta}^\Theta \swich{U_\rho^*\otimes \I},
\end{align} 
where $U_\rho$ is a unitary that diagonalises $\rho$ from the canonical basis and $\Theta$ symbolises the Hadamard (entry-wise) inverse. To show it is equal we need to see that the equation yields the correct coefficients. First, we can remove the enveloping $\swich{U_\rho\otimes\I}\cdot \swich{U_\rho^*\otimes\I}$ and work in the diagonal basis of $\rho$, as done in \cref{thm:JordanInverse}. Then, note that $\swich{U_\rho^*\otimes \I} \omega \swich{U_\rho\otimes \I}$ is just $\omega$ written in the diagonal basis of $\rho$ in the first subsystem and the canonical basis in the second, that is 
\begin{align}
    \swich{\swich{U_\rho^*\otimes \I}\omega \swich{U_\rho\otimes \I}}_{ikj\ell}=\bra{ik}\omega\ket{j\ell}.
\end{align} 
We then invert it element-wise and multiply by half the diagonal element of $\ket{i}$ and $\ket{j}$, that is \begin{align}
    \swich{U_\rho^*\rho U_\rho\star \swich{\swich{U_\rho^*\otimes \I} \omega \swich{U_\rho\otimes \I}}^\Theta}_{ikj\ell} = \fmig\frac{\swich{p_i+p_j}}{\bra{ik}\omega\ket{j\ell}}.
\end{align} 
Now, we only need to invert it element-wise.

Finally, we want to point out that \cref{thm:JordanInverse} recovers Bayes' Theorem in the classical case:

\begin{remark}\label{remark:classicalBayes}
 Let $\rho= \sum_i p_i\ketbra{i}$ and $J_{A\rightarrow B} =\sum_{ij}M_{i\rightarrow j}\ketbra{ij}$, where  $M_{i\rightarrow j}$ is a classical stochastic map. Then $Q= \sum_{ij} M_{i\rightarrow j}p_i\ketbra{ij}=\sum_{ij}r_{ij}\ketbra{ij}$, where $r_{ij}=M_{i\rightarrow j}p_i$ is a joint probability distribution. We can now apply \cref{thm:JordanInverse} considering $B$ the input space. The partial trace will be $\Tr_A\cwich{Q}=\sum_j\swich{\sum_i r_{ij}}\ketbra{j}=\sum_jq_j\ketbra{j}$. $Q$ is already diagonal in a product basis of the required form so we can directly find $J_{A\leftarrow B}=\sum_{ij}r_{ij}/q_j\ketbra{ij}=\sum_{ij}N_{i\leftarrow j}\ketbra{ij}$. Joining everything together we recover Bayes' Theorem: $r_{ij}=N_{i\leftarrow j}q_j=M_{i\rightarrow j}p_i$. 

 Note that, similarly to \cref{thm:JordanInverse}, in the case where $q_j=0$ $N_{i\leftarrow j}$ remains undefined. In the classical case is is possible to choose any values for $N_{i\leftarrow j}$ such that $N$ is a stochastic map since these coefficients do not affect the system, as $N_{i\leftarrow j}q_j=0$ regardless of the chosen value.
\end{remark}

We will also be interested in the cone generated by a set of states over time, $\mc{Q}(\H_A:\H_B)$. A cone is a subset of a real vector space that is closed under product by positive constant and convex combination. 
We consider the cone and not the convex hull because we are interested in studying the positivity of the quantum transport cost, and the cone gives us a better framework to study it later. 
The normalisation condition is linear and therefore much easier to deal with.

\begin{defi}\label{def:coneQStates}
    Let $\H_A,\H_B$ be finite dimensional Hilbert spaces and $\mc{J}(\H_A\rightarrow\H_B)$ the set of {\jamiol} matrices between these two spaces. 
    The cone of states over time is  \begin{equation}\begin{split}
    \hat{\mc{Q}}(\H_A:\H_B)&=\cone{\set{\rho\star J | \,J\in\mc{J}(\H_A\rightarrow\H_B) \text{ and }  \rho\in S(\H_A)}}\\
    &=\cone{\mc{Q}(\H_A:\H_B)}=\cone{\bigcup_{\substack{\rho\in S(\H_A)\\\sigma\in S(\H_B)}}\mc{Q}(\rho,\sigma)}.
\end{split}\end{equation}  
\end{defi}

\subsection{States over multiple times}
\label{subsec:multipleTimes}
If instead of a single channel we have $n$ channels in sequence we can write the space of states over time associated to this process as a recursive hierarchy:
\begin{defi}\label{def:QSHierarchy}
Consider Hilbert spaces $\H_i$, $i\geq 0$. Then
    \begin{align}\begin{split}
        \mc{Q}(\H_0:\cdots:\H_n)&
        =\swich{\mc{Q}(\H_0:\cdots:\H_{n-1})\otimes\I_n}\star\swich{\bigotimes_{j=0}^{n-2}\I_j\otimes\mc{J}(\H_{n-1}\rightarrow\H_{n})}\\
        &\equiv \mc{Q}(\H_0:\cdots:\H_{n-1})\star \mc{J}(\H_{n-1}\rightarrow\H_n).
    \end{split}\end{align}
\end{defi}

By taking partial traces on specific subsystems we can `forget' about the state of the system at that slot, that is that $\Tr_i\cwich{\mc{Q}(\H_0:\cdots:\H_n)}=\mc{Q}(\H_0:\cdots:\H_{i-1}:\H_{i+1}:\cdots:\H_n)$. This is true because given two {\jamiol} matrices $J_{i-1,i},\, J_{i,i+1}$ we can construct a {\jamiol} matrix $\tilde{J}_{i-1,i+1}=\Tr_i\cwich{\swich{J_{i-1,i}\otimes\I_{i+1} }\star\swich{ \I_{i-1}\otimes J_{i,i+1}}}$ such that the associated channels fulfil $\tilde{\cE}_{i-1,i+1}=\cE_{i,i+1}\circ\cE_{i-1,i}$ \cite{chiribella09}.

Given a state over $n$ times, we can also reconstruct the initial and all instantaneous states of the system, as well as the CPTP maps linking different times. Indeed, by tracing out every subsystem except a consecutive pair, we can reduce the problem to the two-time scenario and use \cref{thm:JordanInverse}.


\section{Quantum transport cost and optimal transport}\label{sec:definition}

Before introducing the optimal transport cost between two states, it is useful to define the cost of transporting a single state through a given quantum channel.

\begin{def}\label{def:kappa}
Let $\H_A$, $\H_B$ be Hilbert spaces, let $\rho\in\S(\H_A)$ be a quantum state, and let $\mathcal{E}:\B(\H_A)\to\B(\H_B)$ be a quantum channel.
For a given cost matrix $K\in\B(\H_A\otimes\H_B)$, we define the \emph{quantum transport cost} of $\rho$ under $\mathcal{E}$ as
\begin{align}\label{eq:qChannelCost}
\kappa(\rho,\mathcal{E})=\Tr\cwich{K Q_{\rho,\mathcal{E}}},
\end{align}
where $Q_{\rho,\mathcal{E}}=\rho\star J_\cE$ denotes the state over time (or \emph{coupling}) associated with sending $\rho$ through the channel $\mathcal{E}$, with {\jamiol} operator $J_\cE$.
\end{def}

This quantity measures the expected cost of transporting the state $\rho$ from system $A$ to system $B$ when the process is fixed to be $\mathcal{E}$.
In general, there may exist many channels $\mathcal{E}$ that transform $\rho$ into a target state $\sigma$.
The \emph{quantum optimal transport cost} identifies, among all such admissible transformations, one that achieves the smallest possible cost: 

\begin{defi}
    Let $\rho,\sigma\in\S(\H_{A,B})$ and $K\in \B(\H_A\otimes \H_B)$ a Hermitian operator. The quantum transport cost with cost matrix $K$ between $\rho$ and $\sigma$ is \begin{align}\label{eq:qTransportDef}
        \mc{K}(\rho,\sigma)=\min_{\cE \,|\,Q_{\rho,\cE}\in\mc{Q}(\rho,\sigma)}\kappa(\rho,\cE)=\min_{Q\in\mc{Q}(\rho,\sigma)}\Tr[KQ].
    \end{align}
\end{defi}

Sometimes, when it is clear from the context, we refer to this quantity as just the cost, with the implication that it is quantum, optimal and transport based. Importantly, this definition explicitly depends on the choice of cost matrix $K$, which serves as the analogue of the cost function in classical optimal transport. The relevance of $\mc{K}(\rho,\sigma)$ and the properties it inherits are determined by this choice. 

As we will discuss later in \cref{sec:properties}, the choice of cost matrix $K$ determines the properties of the transport cost. In particular, we are interested in the properties that define a distance (cost $0$ for identical inputs, positivity, triangle inequality and symmetry) as well as some generically useful properties in quantum information (additivity, joint convexity). Ideally, we would be able to find cost matrices $K$ with good properties, notably the weaker ones: cost $0$ for identical inputs and positivity, that allows us to encode some physical property we are interested in considering in a distinguishability measure. In this regard, transport costs differ from traditional distinguishability measures that rely purely on the vector space structure of the states. 

Note that the quantum optimal transport cost can be calculated with an SDP \cite{boyd04,watrous11}. Let, again, $K$ be some cost matrix and $\rho$ and $\sigma$ states. Then the SDP associated to finding the quantum optimal transport cost between $\rho$ and $\sigma$ with cost matrix $K$ is \begin{align}\begin{split}
    \min_J &\quad \quad\; \Tr\cwich{(K\star \rho) J} \\
    \text{s.t.} & \quad\:\left\{\begin{aligned}
        \Tr_B J&=\I\\
        \Tr_A\cwich{\rho J}&=\sigma \label{eq:SDPform}\\
        J^{T_A}&\geq0
       \end{aligned} \right..
\end{split}\end{align} This form can provide some problems computationally due to the partial transpose in the positivity condition. To move these partial transposes to the constants of the problem, we present the SDP with the optimisation over the channel being in Choi form

\begin{align}\begin{split}
    \min_C &\quad \quad\; \Tr\cwich{(K\star \rho)^{T_A} C} \\
    \text{s.t.} & \quad\:\left\{\begin{aligned}
        \Tr_B C&=\I\\
        \Tr_A\cwich{\rho^T C}&=\sigma \label{eq:SDPformApp2}\\
        C&\geq0
       \end{aligned} \right..
\end{split}\end{align} 
and its dual \begin{align}\begin{split}
    \max_{Y_1,Y_2} &\quad \quad\; \Tr\cwich{Y_1}+\Tr\cwich{\sigma Y_2} \\
    \text{s.t.} & \quad\left\{\begin{aligned}
        Y_1\otimes\I+\rho^T\otimes Y_2 &\leq (K\star\rho)^{T_A}\\
        Y_1, Y_2 &\text{ Hermitian}
       \end{aligned} \right..
\end{split}\end{align}

The formulation of the primal in Choi form allows us to see the relation between the primal and dual problems. We use the notation of \cite{watrous11}. In this notation, the primal SDP is given by the triple $(\Phi(X),A,B)=((\Tr_B\cwich{X},\Tr_A\cwich{\rho^T X}),-(K\star \rho)^{T_A},(\I, \sigma))$. Then $\Phi^*((Y_1,Y_2))=Y_1\otimes\I+\rho^T\otimes Y_2$.

The primal expression of the SDP further shows the connection between the coupling and the channel. In fact, the coupling is only implicitly in the SDP through $\Tr\cwich{(K\star \rho) J}=\Tr\cwich{K(\rho\star J)}=\Tr\cwich{K Q}$. We use the {\jamiol} matrix in the SDP instead of the coupling in the SDP because it is unclear how the couplings can be characterised through semidefinite expressions.

\section{Properties}
\label{sec:properties}
In classical probabilistic settings, it is easy and of supreme interest  
to use transport costs to define metrics on probability distributions. 
Recall that a metric on a set $\Omega$ is a function $D:\Omega\times\Omega \rightarrow \R$
(the distance) that satisfies the following properties for all $x,y,z\in\Omega$: 
\begin{align}
  D(x,y)&\geq 0 \\
  D(x,y) &=D(y,x), \\
  D(x,y) &=0 \text{ if and only if } x=y, \\ \label{eq:metricCostZero}
  D(x,z) &\leq D(x,y)+D(y,z).
\end{align}

Note that not all properties described here are equal. While symmetry and the triangle inequality can be nice, a function without these properties can still be a divergence, a very important and useful concept in information theory \cite{ogawa00,tomamichel16}. Thus, finding conditions of $K$ for which the associated quantum optimal transport cost fulfils these properties is an important question.

The main goal of this section is to study sufficient conditions on $K$ for which the associated cost has desirable properties. This is the main theoretical framework and general results of our approach to quantum optimal transport. A better characterisation of which cost matrices induce certain properties in the associated cost would allow us to define physically relevant costs with desirable properties.
Additionally, we examine general mathematical properties of interest in this context, namely convexity properties of the cost under ensemble input states as well as additivity/multiplicity under tensor products. 

In the next section, we will use these results to focus on particular choices of cost functions under further physically relevant constraints, namely unitary invariance of the cost.

\subsection{Cost of identity}

For a quantum transport cost to be reasonable we want that the cost associated with the identity channel is $0$, regardless of the input state. 
We could ask that there exists a state over time such that the cost is $0$ if $\rho=\sigma$, like in \cref{eq:metricCostZero}, but physically it makes sense to ask that doing nothing results in zero cost.
The following result characterises the cost matrices that fulfil this property.

\begin{prop}\label{prop:zerocosts}
Given a finite dimensional Hilbert space $\H$, a cost matrix $K$ assigns cost 0 to the identity map (with any input) if and only if 
\begin{align}\begin{split}
    \Tr_B\cwich{\S\star K}=0,
\end{split}\end{align} 
where $\S$ is the swap operator.
\end{prop}

\begin{proof}
The {\jamiol} operator associated to the identity channel is the swap operator $\S$, clearly from \cref{sec:preliminaries}: $J_{\mathrm{id}}=(\mathrm{id}\otimes \mathrm{id})(\S)=\S$. Now, let $K$ be a matrix such that $\kappa(\rho,\id)=\Tr\cwich{(\rho\star\S)K}=0$ $\,\forall \rho\geq0$. We can transform the left-hand side in the following way, using the definition of the Jordan product, the cyclic property of the trace and properties of partial traces:
\begin{align}
    \Tr\cwich{((\rho\otimes\I)\star \S)K}&=\fmig\Tr\cwich{((\rho\otimes\I) \S+\S (\rho\otimes \I))K} =\fmig\Tr\cwich{(\rho\otimes\I)(\S K+K\S)}\\
    &=\Tr\cwich{\rho\Tr_B\cwich{\S\star K}}=0.
\end{align} 
The set of positive matrices generates the whole space \cite{skowronek09,bakonyi11}, therefore this is equivalent to saying that the Hilbert-Schmidt inner product of $\Tr_B\cwich{\S\star K}$ with all other elements is 0.
Therefore $\Tr_B\cwich{\S\star K}=0$.
The converse is immediate, so the proof is done.
\end{proof}


\subsection{Positivity}
\label{sec:positivity}

A natural requirement for a cost matrix $K$ is that the associated quantum optimal transport cost be nonnegative, \emph{i.e.} 
$\mathcal{K}(\rho, \sigma) \ge 0$ for all $\rho, \sigma \in \mathcal{S}(\mathcal{H}_A, \mathcal{H}_B)$.
Note that in the definition of the cost matrix $K$ is always Hermitian. Because $K$ and the states over time $Q\in\mc{Q}(\H_A:\H_B)$ are Hermitian the cost is a real number. Moreover, since $\Tr\cwich{KQ}=\langle Q,K\rangle_{HS}$, where $\langle \cdot,\cdot\rangle_{HS}$ denotes the Hilbert-Schmidt inner product, the $k$ for which the cost is positive is the dual cone to $\mc{Q}(\H_A:\H_B)$ with respect to the Hilbert-Schmidt inner product. We have no closed-form characterisation of this dual cone, but we can provide some partial results. 
Throughout this section, whenever we refer to duality, we mean it with respect to the Hilbert–Schmidt inner product. 

We will be working with the cone generated by the set of states over time: $\hat{\mc{Q}}(\H_A:\H_B)=\cone{\mc{Q}(\H_A:\H_B)}$, as seen in \cref{def:coneQStates}.

\begin{remark}
The cone can also be obtained by adding an ancillary system $R$ and considering the CPTP maps that take inputs in the composite system $AR$. More precisely, consider finite dimensional Hilbert spaces $\H_A,\H_B,\H_R$ and the following: let 
\begin{align}
        \rho=\sum_t p_t\rho_t\in\S(\H_A),\quad\left\{\cE^{A\rightarrow B}_t\right\},
\end{align} 
where $\cE^{A\rightarrow B}_t $ are quantum channels from $\H_A$ to $\H_B$. 
Now consider the extension to a conditional quantum channel and the following extended state 
\begin{align}
        \rho^{AR}=\sum_t p_t\rho^A_t\otimes\ketbra{t}^R, \quad \cE^{AR\rightarrow B}(X_A\otimes X_R)=\sum_t \cE_t^{A\rightarrow B}(X_A)\expval{X_R}{t},   
\end{align} 
where the definition of $\cE^{AR\rightarrow B}$ is extended to general elements by linearity. The {\jamiol} operator associated to $\cE^{AR\rightarrow B}$ is $J^{AR\rightarrow B}=\sum_t J^{A\rightarrow B}_t\otimes\ketbra{t}^R$. Now 
\begin{align}
        J^{AR\rightarrow B}\star\rho^{AR}=\sum_tp_tJ_t^{A\rightarrow B}\star\rho^A_t\otimes\ketbra{t}^R,
\end{align} 
and the partial trace (removing $R$) of this state over time is 
\begin{align}
        \Tr_R\cwich{J^{AR\rightarrow B}\star\rho^{AR}}=\sum_tp_tJ_t^{A\rightarrow B}\star\rho^A_t,
\end{align} 
which is an arbitrary convex combination of states over time on $\mc{Q}(\H_A:\H_B)$. 
\end{remark}

First, we want to characterise the dual cone to the set of {\jamiol} matrices $\mc{J}(\H_A\rightarrow \H_B)^*$. For this, we need some technical results about convex cones, starting with the precise definition of dual cone in inner product vector spaces.

\begin{defi}\label{def:dualCone}
Let $X\subseteq V$ be a subset of a real inner product vector space $V$. The dual cone to $X$ with respect to the inner product is a cone in $V$, denoted by $X^*$, and defined as \begin{equation}
    X^*=\set{v\in V \,|\, \left\langle v|x\right\rangle\geq0 \,\forall x\in X}.
\end{equation} 
\end{defi}

In the context of this work, $V$ is the Hermitian subset of $\B(\H_A\otimes\H_B)$ equipped with the Hilbert-Schmidt inner product. Note that with no further requirements for $X$, $X^*$ fulfils the definition of cone.  

\begin{lemma}\label{lem:dualSumIntersec}
Let $I$ be an index set and $\set{\mc{C}_i}_{i\in I}$ a set of convex cones. Then 
\begin{align}
    \bigcap_{i\in I}\mc{C}_i^*=\swich{\sum_{i\in I}\mc{C}_i}^*.
\end{align}
\end{lemma}

\begin{proof}
We can show the equality directly. Let $x\in\bigcap_{i\in I}\mc{C}_i^*$. Then,
\begin{align}\begin{split}
     x\in\mc{C}_i^*\;\forall i\in I\quad&\Leftrightarrow\quad \innerproduct{x}{c_i}\geq0\;\forall c_i\in\mc{C}_i\;\forall i\in I \quad\Leftrightarrow\quad \sum_{i\in I}\innerproduct{x}{c_i}\geq0\;\forall c_i\in\mc{C}_i \\
     &\Leftrightarrow\quad \innerproduct{x}{\sum_{i\in I}c_i}\geq0\; \forall c_i\in\mc{C}_i \quad\Leftrightarrow\quad x\in\swich{\sum_{i\in I} \mc{C}_i}^*.
\end{split}\end{align} 
\end{proof}

Recall that a convex cone $C\subseteq V$ is called \emph{pointed} if for all 
nonzero $x\in C$, $-x\notin C$. 

\begin{lemma}\label{lem:pointedCone}
A cone $C$ is pointed if and only if there exists an element $f$ in the dual space 
such that $f(x)>0$ for all nonzero $x\in C$.
\end{lemma}

\begin{proof}
Let $x,-x\in C$ and $f\in C^*$ such that $f(x)>0$ for all nonzero $x\in C$. 
Because $f$ is linear $f(-x)=-f(x)<0$, which is a contradiction. 

Conversely, let $C$ be pointed, then $(C\setminus\{0\})\cap((-C)\setminus\{0\})=\emptyset$. 
By the Hahn-Banach Theorem (as formulated in \cite[Theorem 7.8.4]{narici10}), there exists a linear map such that $f(x)>0$ for all $x\in C\setminus\{0\}$.
\end{proof}

\begin{lemma}\label{lem:dualPropI}
Consider the convex cone $\mathcal{C}_2=\{C\in\B(\H_A\otimes\H_B)\;|\; \Tr_B\cwich{C}\propto_{\C}\I\}$. Its dual is 
\begin{align}
    \mathcal{C}_2^*=\set{A\otimes\I\in\B(\H_A\otimes\H_B)\;| \; A\in\B(\H_A),\I\in\B(\H_B),\; \Tr\cwich{A}=0}.
\end{align}
\end{lemma}

\begin{proof}
Let us call this set $\mathcal{A}=\set{A\otimes\I\in\B(\H_A\otimes\H_B)\;| \;A\in\B(\H_A),\; \Tr\cwich{A}=0}$. The following calculation shows that $\mathcal{A}\subseteq\mathcal{C}_2^*$: let $A\otimes\I\in\mathcal{A}$ and $C\in\mathcal{C}_2$, then: 

\begin{align}
\Tr\cwich{\swich{A\otimes\I}C}=\Tr_A\cwich{\Tr_B\cwich{\swich{A\otimes\I}C}}=\Tr\cwich{A\Tr_B\cwich{\I C}}=z\Tr\cwich{A\I}=0.
\end{align} 

To see that they are equal, note that $\mathcal{C}_2$ (and thus the orthogonal $\mathcal{C}_2^*$ \cite{bakonyi11}) and $\mathcal{A}$ are real subspaces of the real vector space $\B(\H_A\otimes\H_B)$. We will calculate the dimension of each and see they are the same. The real dimension of $\mathcal{A}$ is just the real dimension of $\B(\H)$ minus the dimension subtracted by the two real (one complex) linear conditions $\Tr\cwich{A}=0$. That is \begin{align}
    \dim \mathcal{A}=\dim \B(\H) -2=2d_A^2-2.
\end{align} 
To find the dimension of $\mathcal{C}^*_2$, we first find the dimension of $\mathcal{C}_2$. Recall that this set is defined by the condition $\Tr_B\cwich{C}\propto_{\C}\I$. This corresponds to $2d_A(d_A-1)$ equations (real and imaginary parts of non diagonal terms equal to 0) plus $2(d_A-1)$. That is because the condition is proportionality, not equality, so we first fix the real and imaginary components of the first diagonal element and then every other diagonal element will have to have the same real and imaginary components, for a total of $2(d_A-1)$. Thus the dimension is 
\begin{align}\begin{split}
    \dim \mathcal{C}_2&=\dim \B(\H_A\otimes\H_B)-2d_A(d_A-1)-2d_A+2\\
    &=\dim \B(\H_A\otimes\H_B)-2d_A^2+2d_A-2d_A+2\\
    &=\dim \B(\H_A\otimes\H_B)-2d_A^2+2.
\end{split}\end{align} 
The dimension of the orthogonal complement is the dimension of the total space minus this, thus 
\begin{align}
    \dim\mathcal{C}^*_2=\dim \B(\H_A\otimes\H_B)-\dim \B(\H_A\otimes\H_B)+2d_A^2-2=2d_A^2-2.
\end{align} 

Since this two sets $\mathcal{A}$ and $\mathcal{C}^*_2$ are real subspaces of the same dimension and $\mathcal{A}\subseteq\mathcal{C}^*_2$, they are the same:
\begin{align}
    \mathcal{C}^*_2=\mathcal{A}=\set{A\otimes\I\in\B(\H_A\otimes\H_B)\;| \;A\in\B(\H_A),\; \Tr\cwich{A}=0}.
\end{align}
\end{proof}

\begin{lemma}\label{lem:PTstrongDuality}
The partial transpose map, denoted here by $T_A(\cdot)$, fulfils the following: 
\begin{align}\begin{split}
 \Tr_A\cwich{T_A(K)C} &= \Tr_A\cwich{K\;T_A(C)}  \\
 T_A(\Tr_B\cwich{K\; C}) &= \Tr_B\cwich{T_A(C)\;T_A(K)}\\
 T_A((\rho\otimes \I)C)&=T_A(C)(\rho^T\otimes \I)
\end{split} 
  \quad\forall K,C\in\B(\H_A\otimes\H_B), \;\rho\in\B(\H_A) .
\end{align} 
Moreover, the partial transpose is self-adjoint with respect to the Hilbert-Schmidt inner product.
\end{lemma}

\begin{proof}
Recall that the transpose is a basis dependent operation. These two properties are trivial to check if we expand the equations in a product basis that includes the basis over which we are transposing. Alternatively, we can use tensor network notation \cite{biamonte20} as shown in \cref{fig:PTraceCommuting}.

\begin{figure}[ht]
    \centering
      \includegraphics[width=0.6\textwidth]{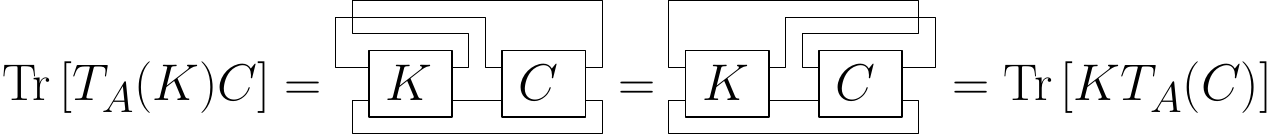} \\
      \vspace{2mm}
      \includegraphics[width=0.7\textwidth]{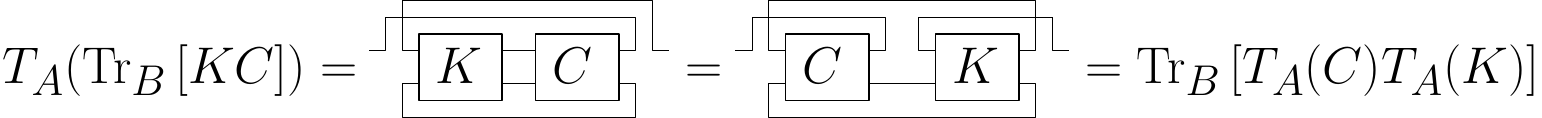} \\
      \vspace{2mm}
      \includegraphics[width=0.6\textwidth]{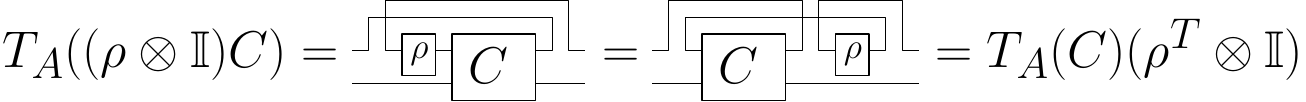}
      \caption{Proofs of the expressions in \cref{lem:PTstrongDuality} using tensor network notation.} 
      \label{fig:PTraceCommuting}
\end{figure}

To see that the partial transpose is self adjoint, 
apply the first equation to $K^*$ and take the trace on both sides of the equation: 
\begin{equation}\begin{split}
   \Tr_B \Tr_A\cwich{T_A(K^*)C} = \Tr_B\Tr_A\cwich{K^*\;T_A(C)}\\
   \Leftrightarrow \Tr\cwich{T_A(K^*)C}=\Tr\cwich{K^*\;T_A(C)}\\
   \Leftrightarrow \langle T_A(K^*) ,C \rangle_{HS} = \langle K^* ,T_A(C) \rangle_{HS}.
\end{split}\end{equation}
\end{proof}

\begin{lemma}\label{lem:linearmapdual}
Let $K$ be a convex cone and $A$ an invertible linear map. Then, 
\begin{align}
    A(K)^*&=(A^*)^{-1}(K^*).
\end{align}
\end{lemma}

\begin{proof}
Let $x\in A(K)^*$. Then, 
\begin{align}\begin{split}
    x\in A(K)^* \quad&\Leftrightarrow\quad \sprod{x}{A(y)}\geq0 \:\:\:\forall y\in K \quad\Leftrightarrow\quad \sprod{A^*(x)}{y}\geq0 \:\:\:\forall y\in K \quad\\
    &\Leftrightarrow\quad A^*(x)\in K^* \quad\Leftrightarrow\quad x\in (A^*)^{-1}(K^*).
\end{split}\end{align} 
\end{proof}

\begin{coro}
The dual of the cone of {\jamiol} operators is the partial transpose of the dual 
cone of the Choi operators, denoted by $\mc{C}$. In other words,
\begin{align}
    \mathcal{J}(\H_A\rightarrow \H_B)^* 
    = T_A(\mathcal{C(\H_A\rightarrow \H_B)})^*
    = T_A(\mathcal{C}(\H_A\rightarrow \H_B)^*).
\end{align}
\end{coro}

\begin{proof}
Note that the partial transpose is self adjoint from \cref{lem:PTstrongDuality} and self inverse and apply the previous \cref{lem:linearmapdual}.
\end{proof}

First, we can show using \cref{lem:pointedCone} that the cones $\hat{\mc{Q}}(\H_A:\H_B)$ and $\hat{\mc{Q}}(\H_A:\H_B)^*$ are pointed and spanning. A cone $C\subset V$ is spanning if $C+(-C)=V$ \cite{boyd04}. 

\begin{prop}\label{prop:QPointedSpanning}
The cone of states over time $\hat{\mc{Q}}(\H_A:\H_B)$ and its dual 
$\hat{\mc{Q}}(\H_A:\H_B)^*$  are pointed, spanning cones. 
\end{prop}

\begin{proof}
    First, we show that $\hat{\mc{Q}}(\H_A:\H_B)$ is pointed and spanning. By definition of the elements $Q\in\hat{\mc{Q}}(\H_A:\H_B)$, $\Tr\cwich{\I Q}=\Tr\cwich{\rho\star J}=\Tr\cwich{\rho}=1>0$. By \cref{lem:pointedCone}, $\mc{Q}(\H_A:\H_B)$ is pointed. 
    The cone $\hat{\mc{Q}}(\H_A:\H_B)$ is spanning because the set of product states $\{\rho\otimes\sigma| \rho,\sigma\in S(\H)\}$ is contained in $\hat{\mc{Q}}(\H_A:\H_B)$. That is because the {\jamiol} matrix of the replacement channel is $\I\otimes \sigma$. This set is spanning so as its superset $\hat{\mc{Q}}(\H_A:\H_B)$ is also spanning.
    
    The properties of pointed and spanning are such that if the primal cone has one, the dual has the other \cite{boyd04}. As we just showed that $\hat{\mc{Q}}(\H_A:\H_B)$ is pointed and spanning, its dual is also pointed and spanning.
\end{proof}

We are mostly interested in the fact that $\hat{\mc{Q}}(\H_A:\H_B)^*$ is spanning from \cref{prop:QPointedSpanning}. This shows that our search for the cone of cost matrices with positive associated costs is not futile since the set of matrices with this property is spanning.

\begin{prop}\label{prop:dualChoi}
Let $\mathcal{C}\in\B(\H_A\otimes\H_B)$ be the minimal cone that contains the Choi matrices.\footnote{Through the Choi isomorphism this would correspond to CP and trace \emph{scaling} (by a real positive constant, instead of trace preserving) maps.} Then, 
\begin{align}\begin{split}
    \mathcal{C}^*&=\overline{\B_+(\H_A\otimes\H_B)+\set{A\otimes\I\in\B(\H_A\otimes\H_B)\;| \;A\in\B(\H_A),\; \Tr\cwich{A}=0}}\\
    &=\overline{\set{\omega+A\otimes\I \in\B(\H_A\otimes\H_B)\;|\; \omega\in\B_+(\H_A\otimes\H_B),\;\Tr\cwich{A}=0}} \\
 &=\set{\omega+A\otimes\I \in\B(\H_A\otimes\H_B)\;|\; \omega\in\B_+(\H_A\otimes\H_B),\;\Tr\cwich{A}=0}.
\end{split}\end{align}
\end{prop}

\begin{proof}
Consider the following: 
\begin{align}
    \mathcal{C}_1&=\B_+(\H_A\otimes\H_B), \\
    \mathcal{C}_2&=\{C\in\B(\H_A\otimes\H_B)\;|\; \Tr_B\cwich{C}\propto_{\C}\I\}.
\end{align} 
These two are closed cones and 
\begin{align}
    \mathcal{C}=\mathcal{C}_1\cap\mathcal{C}_2.
\end{align} 
Moreover (the cone of psd matrices is self dual \cite{bakonyi11} and \cref{lem:dualPropI}):  
\begin{align}
   \mathcal{C}_1^*&= \mathcal{C}_1=\B_+(\H_A\otimes\H_B), \\
    \mathcal{C}_2^*&= \set{A\otimes\I\in\B(\H_A\otimes\H_B)\;| \;A\in\B(\H_A),\; \Tr\cwich{A}=0}.
\end{align} 
Now, we can use Lemma~\ref{lem:dualSumIntersec}, setting $I=\{1,2\}$ and the duals in the theorem, to find the dual of $\mathcal{C}$: 
\begin{align}\begin{split}
    \mathcal{C}^*&=(\mathcal{C}_1\cap\mathcal{C}_2)^*=(\overline{\mathcal{C}_1}\cap\overline{\mathcal{C}_2})^*=\overline{\mathcal{C}_1^*+\mathcal{C}_2^*}\\
    &=\overline{\B_+(\H_A\otimes\H_B)+\set{A\otimes\I\in\B(\H_A\otimes\H_B)\;| \;A\in\B(\H_A),\; \Tr\cwich{A}=0}},
\end{split}\end{align} 
where we used $\mathcal{C}=\mathcal{C}_1\cap\mathcal{C}_2$ first; the closeness of $\mathcal{C}_1$ and $\mathcal{C}_2$ second, then \cref{lem:dualSumIntersec}; and finally the duals of $\mathcal{C}_1$ and $\mathcal{C}_2$. To see that the set \begin{equation}
\set{\omega+A\otimes\I \in\B(\H_A\otimes\H_B)\;|\; \omega\in\B_+(\H_A\otimes\H_B),\;\Tr\cwich{A}=0}
\end{equation} is closed, note that both $\mc{C}_1^*$ and $\mc{C}_2^*$ are closed due to being dual cones. Moreover, let $\omega\in\mc{C}_1$ and $A\otimes\I\in\mc{C}_2$ have norm 1 (in the norm induced by the inner product). First note that the minimum eigenvalue of $A$, is lower bounded because $1=\norm{A\otimes\I}=\norm{A}\norm{I}=\sqrt{d}\norm{A}$ and $\frac{1}{\sqrt{d}}=\norm{A}=\sqrt{\sum_i\lambda^2_i}$. Due to $\Tr\cwich{A}=0$, the smallest possible eigenvalue of $A$, $\lambda_{\min}$ appears when $A$ has two eigenvalues, $\lambda_{\min}$ and $-\lambda_{\min}$. Then \begin{equation}
        \frac{1}{\sqrt{d}}=\norm{A}=\sqrt{\sum_i{\lambda_i^2}}=\sqrt{2\lambda_{\min}^2},
\end{equation} so that $\lambda_{\min}=-\frac{1}{\sqrt{2d}}.$ Now we can calculate the lower bound for the norm of the sum of these two operators: 

\begin{align}
    \begin{split}
\norm{w+A\otimes\I}&=\sqrt{\Tr{\cwich{\swich{\omega+A\otimes\I}^*\swich{\omega+A\otimes\I}}}}=\sqrt{\Tr\cwich{\omega^2}+\Tr\cwich{\abs{A}^2\otimes\I}+2\Tr\cwich{\omega\swich{A\otimes\I}}}  \\
&=\sqrt{2+2\Tr\cwich{\Tr_B\cwich{\omega}\swich{A\otimes\I}}}\geq\sqrt{2-2\lambda_{\min}}=\sqrt{2\frac{\sqrt{2d}-1}{\sqrt{2d}}}>0.
\end{split}
\end{align} By \cite{beutner07}, the sum is closed, finishing the proof.
\end{proof}

For the following theorem we use the notation $D_s$, where $s\in\R_+^n$ is a sequence on positive numbers, to indicate the diagonal matrix with $s$ in its diagonal. 

\begin{thm}\label{thm:dualConeStotes}
The dual to the set of states over time for finite dimensional Hilbert spaces $\H_A,\H_B$, $\hat{\mc{Q}}(\H_A:\H_B)^*$ can be expressed as 
\begin{align}
    \hat{\mc{Q}}(\H_A:\H_B)^*=\bigcap_{U\in{\mc{U}(\H_A)}}\swich{U\otimes \I}\swich{\bigcap_{s\in\R_+^{d_A}}\varphi_{D_s}^{-1}(\mc{J}(\H_A\rightarrow\H_B)^*)}\swich{U^*\otimes \I},
\end{align} 
where $\varphi_\rho(X)=\rho\star X$ and $\mc{J}(\H_A\rightarrow\H_B)^*$ is the dual to the set of {\jamiol} matrices.
\end{thm}  

\begin{proof}
For simplicity, we ignore the specific Hilbert space dependencies. We want to study $\hat{\mc{Q}}^*=\swich{\sum_\rho\swich{\rho\star\mc{J}}}^*$. Given a non faithful $\rho$ and a fixed $J\in\mathcal{J}$, we can arbitrarily approximate $\rho\star J$ with faithful $\rho_\varepsilon$ as \begin{equation}
\swich{\frac{\rho+\varepsilon\I}{\Tr\cwich{\rho+\varepsilon\I}}}\star J,
\end{equation} $\varepsilon>0$.  Let $K$ be such that $\Tr{K(\rho\star J)}<0$. Then, by continuity, $\exists \varepsilon_0>0$ such that $\Tr{K(\rho_\varepsilon\star J)}<0$ for $0<\varepsilon<\varepsilon_0$.\footnote{Note this argument is equivalent to removing part of the boundary of the set $\hat{\mc{Q}}$ and then noting that the dual cone is always the same for a set and its interior.} Thus we can consider only faithful $\rho$ in the definition of $\mc{Q}^*$. This allows us to take the inverse without problems.

Start with the definition of $\hat{\mc{Q}}$, then apply \cref{lem:dualSumIntersec} and \cref{lem:linearmapdual}: 
\begin{align}
    \hat{\mc{Q}}^*=\swich{\sum_\rho\swich{\rho\star\mc{J}}}^*=\swich{\sum_\rho\varphi_\rho(\mc{J})}^*=\bigcap_{\rho}\varphi_\rho(\mc{J})^* =\bigcap_{\rho}\varphi_\rho^{-1}(\mc{J}^*).
\end{align} 
Note that we can use \cref{lem:linearmapdual} because for a fixed $\rho$, $\varphi_\rho$ is self dual and has linear inverse, as can be seen from the statement of the inverse in \cref{thm:JordanInverse}. From here, realise that choosing a state $\rho$ is equivalent to choosing a spectrum and a basis or, equivalently, a spectrum $s\in\R^n_+$ and a unitary of $\mc{U}(n)$; such that $\rho=U_\rho D_{s_\rho}U^*_\rho$. Moreover, $\mc{J}^*$ is invariant under local unitaries, thus 
\begin{align}\begin{split}
    \varphi_\rho^{-1}(\mc{J}^*) &=\swich{U_\rho\otimes \I}\swich{U_\rho^*\rho U_\rho\star\swich{\swich{U_\rho^*\otimes \I}\mc{J}^*  \swich{U_\rho\otimes \I}}^\Theta}^\Theta \swich{U_\rho^*\otimes \I} \\
    &=\swich{U_\rho\otimes \I}\swich{D_{s_\rho}\star\swich{\mc{J}^* }^\Theta}^\Theta \swich{U_\rho^*\otimes \I} 
    =\swich{U_{\rho}\otimes \I}\varphi_{D_{s_{\rho}}}^{-1}(\mc{J}^*)\swich{U^*_\rho\otimes \I}.
\end{split}\end{align} 
And we can insert this result into the expression of $\hat{\mc{Q}}^*$ to obtain that 
\begin{align}\begin{split}
    \hat{\mc{Q}}^*&=\bigcap_{\rho}\varphi_\rho^{-1}(\mc{J}^*)=\bigcap_{U\in U(\H_A)}\bigcap_{s\in\R^{d_A}_+}\swich{U\otimes \I}\varphi^{-1}_{D_s}(\mc{J}^*)\swich{U^* \otimes \I}\\
    &=\bigcap_{U\in U(\H_A)}\swich{U\otimes \I}\swich{\bigcap_{s\in\R^{d_A}_+}\varphi^{-1}_{D_s}(\mc{J}^*)}\swich{U^* \otimes \I},
\end{split}\end{align} 
which is the local unitarily invariant subset of $\bigcap_{s\in\R^{d_A}_+}\varphi^{-1}_{D_s}(\mc{J}^*)$.
\end{proof}

\cref{thm:dualConeStotes} provides a partial characterisation of the dual to the set of states over time as the local unitary invariant subset of the intersection of the images of $\mc{J}^*$ with the inverse of the Jordan product with diagonal states (in a chosen basis). Even though this inverse was discussed in \cref{thm:JordanInverse}, we have not been able to use it to obtain a complete and succinct characterization of the dual cone.

Finally, we can assume both non-negativity of the cost and zero cost for the identity channel to obtain the following results in the case where $\H_A=\H_B$: 

\begin{thm}\label{thm:JdualIzero}
Let $\H$ be a finite dimensional Hilbert space. Then $K \in\mathcal{J}(\H\rightarrow\H)^*\cap\set{C\in\B(\H\otimes\H)\;|\;\Tr_B\cwich{\S\star C}=0}$ if and only if 
\begin{align}
    K=T_A\swich{\omega}-\swich{\Tr_B\cwich{\S\star T_A\swich{\omega}}\otimes\I},\quad \omega\geq0,\;\;\omega\perp\ketbra{\Phi_+}.
\end{align} 
\end{thm}

In the previous theorem, $K$ is a matrix that is dual to the {\jamiol} matrices and generates cost $0$ for the identity (see \cref{prop:zerocosts}, \cref{prop:dualChoi} and \cref{lem:PTstrongDuality}). 

\begin{proof}
Similarly to before, we ignore the Hilbert space dependencies for the proof. Note that even though the identity $T(\mathcal{C}_1\cap\mathcal{C}_2)=T(\mathcal{C}_1)\cap T(\mathcal{C}_2)$ for a linear map $T$ and convex cones $\mathcal{C}_1$, $\mathcal{C}_2$, is false in general, it is true for the partial transpose because that is an invertible map. Thus we can transform the target set as follows: 
\begin{align}\begin{split}
    \mathcal{J}^*\cap&\set{C\in\B(\H\otimes\H)\;|\;\Tr_B\cwich{\S\star C}=0} \\
    &=T_A(T_A(\mathcal{J}^*\cap\set{C\in\B(\H\otimes\H)\;|\;\Tr_B\cwich{\S\star C}=0}))\\
    &=T_A(T_A(\mathcal{J}^*)\cap T_A(\set{C\in\B(\H\otimes\H)\;|\;\Tr_B\cwich{\S\star C}=0}))\\
    &=T_A(\mathcal{C}^*\cap\set{C\in\B(\H\otimes\H)\;|\;\Tr_B\cwich{\S\star T_A(C)}=0})\\
    &=T_A(\mathcal{C}^*\cap\set{C\in\B(\H\otimes\H)\;|\;T_A(\Tr_B\cwich{\S\star T_A(C)})=T_A(0)})\\
    &=T_A(\mathcal{C}^*\cap\set{C\in\B(\H\otimes\H)\;|\;\Tr_B\cwich{T_A(\S)\star C}=0})\\
    &=T_A(\mathcal{C}^*\cap\set{C\in\B(\H\otimes\H)\;|\;\Tr_B\cwich{\ketbra{\Phi_+}\star C}=0}),
\end{split}\end{align} 
where we used \cref{lem:PTstrongDuality}.

Now consider an element of $\mathcal{C}^*$, that is a $K=\omega+A\otimes\I$, where $\omega\geq0$ and $\Tr\cwich{A}=0$. We can now plug this expression in the equation that defines the other set of the intersection: \begin{align}
    0 = \Tr_B\cwich{\ketbra{\Phi_+}\star K} 
      = \Tr_B\cwich{\ketbra{\Phi_+}\star (\omega+A\otimes\I)}
      = \Tr_B\cwich{\ketbra{\Phi_+}\star \omega}+A, 
\end{align} 
thus $A=-\Tr_B\cwich{\ketbra{\Phi_+}\star \omega}$. Moreover if we take the trace of this expression, since $\Tr\cwich{A}=0$, we find that $\expval{\omega}{\Phi_+}=0$, \emph{i.e.} $\omega\perp\ketbra{\Phi_+}$. Now, the initial set is the set defined by the partial transpose of this elements, that is
\begin{align}\begin{split}
    \mathcal{J}^*&\cap\set{C\in\B(\H\otimes\H)\;|\;\Tr_B\cwich{\S\star C}=0} \\ 
    &=T_A(\set{\omega-\Tr_B\cwich{\ketbra{\Phi_+}\star \omega}\otimes\I\in\B(\H^2)\;|\; \omega\geq0, \;\omega\perp\ketbra{\Phi_+}})\\
    &=\set{T_A(\omega)-T_A(\Tr_B\cwich{\ketbra{\Phi_+}\star \omega}\otimes\I)\in\B(\H^2)\;|\; \omega\geq0, \;\omega\perp\ketbra{\Phi_+}}\\
    &=\set{T_A(\omega)-\Tr_B\cwich{T_A(\ketbra{\Phi_+})\star T_A(\omega)}\otimes\I\in\B(\H^2)\;|\; \omega\geq0, \;\omega\perp\ketbra{\Phi_+}}\\
    &=\set{T_A(\omega)-\Tr_B\cwich{\S\star T_A(\omega)}\otimes\I\in\B(\H^2)\;|\; \omega\geq0, \;\omega\perp\ketbra{\Phi_+}}.
\end{split}\end{align}
\end{proof}

Using a variation of \cref{lem:linearmapdual} we can show how the dual cone to the cone of states over multiple times behaves under partial traces.

\begin{thm}
Let $\hat{\mc{Q}}(\H_0:\cdots:\H_n)$ be defined as 
\begin{equation}
    \hat{\mc{Q}}(\H_0:\cdots:\H_n)=\cone{\hat{\mc{Q}}(\H_0:\cdots:\H_{n-1})\star\mc{J}(\H_{n-1}\rightarrow\H_n)}.
\end{equation} 
The dual of this hierarchy fulfils 
\begin{equation}
    \Tr_{i}\cwich{\hat{\mc{Q}}(\H_0:\cdots:\H_n)^*}\supseteq\Tr_i\cwich{\hat{\mc{Q}}(\H_0:\cdots:\H_n)}^*.
\end{equation}
\end{thm}

\begin{proof}
We can show this for general cones using the proof of \cref{lem:linearmapdual}. Let $K\subseteq\B(\H_A\otimes\H_B)$ be a cone and, then \begin{align}\begin{split}
    x\in \Tr_B(K)^* \quad&\Leftrightarrow\quad \sprod{x}{\Tr_B(y)}\geq0 \:\:\:\forall y\in K \quad\Leftrightarrow\quad \sprod{x\otimes \I}{y}\geq0 \:\:\:\forall y\in K \quad\\
    &\Leftrightarrow\quad x\otimes\I\in K^* \quad\Rightarrow\quad x\in \Tr_B(K^*).
\end{split}\end{align} 

We can now set $\H_A=\H_0\otimes\cdots\H_{i-1}\otimes\H_{i+1}\otimes\cdots\H_n$, $\H_B=\H_i$ and $K=\hat{\mc{Q}}(\H_0:\cdots:\H_n)$ to complete the proof. 
\end{proof}

\subsection{Symmetry}

We call an optimal quantum cost symmetric if 
\begin{equation}
    \mathcal{K}(\rho,\sigma)=\mathcal{K}(\sigma,\rho),\quad \forall \rho,\sigma\in\S(\H_{A,B}).
\end{equation}

A naive guess or first thought might suggest that symmetry of the cost matrix under input–output exchange, i.e., $\S K \S = K$, guarantees symmetry of the associated cost. We will see that this not the case. To understand this, we first address another natural question,  namely whether the set of states over time itself is symmetric. In other words, we ask whether swapping the input and output space in the coupling always produces a valid state over time for the reversed pair of states, that is, whether $\S\mc{Q}(\rho,\sigma)\S=\mc{Q}(\sigma,\rho)$.

There is a clear asymmetry in the case when, for example, $\rho$ is full rank and $\sigma$ is pure, where a single channel that brings $\rho$ to $\sigma$ exists. Despite that, as seen in \cref{thm:JordanInverse}, if the input state is not full rank a single state over time can correspond to multiple channels. Therefore, in the case where $\sigma$ is pure, it could be that the single channel that brings $\rho$ to $\sigma$ inverts into all the channels that bring $\sigma$ to $\rho$. 

In the following example we perform specific calculations for the case where both states are qubits. We observe, in examples \textit{iii), iv)}, that some states over time cannot be time inverted, thus strongly suggesting that $\S K\S =K$ is not a sufficient condition on the cost matrix to give rise to symmetry of the cost.

\begin{example}\begin{enumerate}[i)]
    
    \item \textbf{Replacement channel:} Let $\rho, \sigma$ be any states and let $J_{A\rightarrow B}$ be the {\jamiol} matrix associated to the constant channel $\cE(\rho)=\Tr(\rho)\sigma$, that is $J_{A\rightarrow B}=\I\otimes\sigma$. The associated state over time is $Q=\rho\star J_{A\rightarrow B}=\rho\otimes\sigma$. From the symmetry of the state over time we can see immediately that we can obtain the same result with $(\sigma,J_{A\leftarrow B}=\rho\otimes\I)$.

    \item \textbf{Identity channel:} Let $\rho$ be a qubit state with eigenvalues $\{p,1-p\}$ and $J_{A\rightarrow B}$ be the {\jamiol} matrix associated to the identity channel, $\S$. Then the associated state over time $Q$ is, in (the tensor basis generated by) the diagonal basis of $\rho$, 
    \begin{align}
        Q=\rho\star\S=\begin{bmatrix}
            p & 0 & 0 & 0\\
            0 & 0 & \frac{1}{2} & 0 \\
            0 & \fmig & 0 & 0 \\
            0 & 0 & 0 & 1-p
        \end{bmatrix}.
    \end{align} 
    Similarly to before, the symmetry (under subsystem swap) allows us to easily show that the pair $(\sigma=\rho,J_{A\leftarrow B}=\S)$ yields the same state over time.

    \item \textbf{Depolarising channel:} Let the initial state be a  pure state, without loss of generality (w.l.o.g.), we will set $\rho=\ketbra{0}$. Let $J_{A\rightarrow B}$ be the {\jamiol} matrix associated to the depolarising channel $\cE(\rho)=(1-p)\rho +p\Tr(\rho)\frac{\I}{2}$, that is $J_{A\rightarrow B}=(1-p)\S+\frac{p}{2}\I$. The resulting state over time is 
    \begin{align}
        Q=\rho\star J_{A\rightarrow B}=\fmig\begin{bmatrix}
            2-p & 0 & 0 & 0 \\
            0 & p & 1-p & 0 \\
            0 & 1-p & 0 & 0 \\
            0 & 0 & 0 & 0 
        \end{bmatrix}.
    \end{align} 
    From this channel, $\sigma=\cE(\rho)=\fmig(2-p)\ketbra{0}+\frac{p}{2}\ketbra{1}$. Applying \cref{thm:JordanInverse} to $Q$ yields \begin{align}
        J_{A\leftarrow B}^{T_B}=\begin{bmatrix}
            1 & 0 & 0 & 1-p \\
            0 & 1 & 0 & 0 \\
            0 & 0 & 0 & 0 \\
            1-p & 0 & 0 & 0 
        \end{bmatrix},
    \end{align} 
    which can be shown through Sylvester's criterion to be not psd by taking the principal minor with $I=\{1,4\}$ if $p\neq1$. If $p=1$, the depolarising channel becomes a replacement channel which we have seen is reversible.

    \item \textbf{Dephasing channel:} Let $\rho=\ketbra{+}$ and $J_{A\rightarrow B}$ be the {\jamiol} matrix associated to the dephasing channel $\cE(\rho)=p\rho+(1-p)\sigma_z\rho\sigma_z$ for $p\in(0,1)$, that is 
    \begin{align}
        J_{A\rightarrow B}=\begin{bmatrix}
            1 & 0 & 0 & 0 \\
            0 & 0 & \frac{2p-1}{2} & 0 \\
            0 & \frac{2p-1}{2} & 0 & 0 \\
            0 & 0 & 0 & 1
        \end{bmatrix}.
    \end{align} 
    Note that $\sigma=\cE(\rho)=\frac{1}{2}(\I+(2p-1)\sigma_x)$, which has rank 2 for $p\in(0,1)$. We can now calculate the associated state over time 
    \begin{align}
        Q=\rho\star J_{A\rightarrow B}=\frac{1}{8}\begin{bmatrix}
            4 & 2p-1 & 2 & 0 \\
            2p-1 & 0 & 2(2p-1) & 2 \\
            2 & 2(2p-1) & 0 & 2p-1 \\
            0 & 2 & 2p-1 & 4
        \end{bmatrix}.
    \end{align} 
    We can now calculate $J_{A\leftarrow B}$ from \cref{thm:JordanInverse},\footnote{Note that even though $\rho$ has rank 1, $\sigma$ has rank 2 and therefore allows us to uniquely apply \cref{thm:JordanInverse}.} which yields 
    \begin{align}
        J_{A\leftarrow B} = \frac{1}{4}\begin{bmatrix}
            4 & 0 & 2 & 1-2p \\
            0 & 0 & 2p-1 & 2 \\
            2 & 2p-1 & 0 & 0 \\
            1-2p & 2 & 0 & 4
        \end{bmatrix}.
    \end{align} 
    This matrix is clearly not psd under partial transposition of $B$ since the principal minor $[J_{A\leftarrow B}]_{\{1,3\}}$ (which is unaffected by the partial transposition) has negative determinant, thus the matrix is not psd from Sylvester's criterion. For example, when $p=\frac{1}{2}$, the eigenvalues of $J_{A\leftarrow B}^{T_A}$ are $\{\frac{1}{2}(1\pm\sqrt{2})\}$.

    \item \textbf{Measure-and-prepare channel:} Let $\rho\in\S(\H_A)$ and $J_{A\rightarrow B}$ be the {\jamiol} matrix of a measure-and-prepare channel, that is a channel of the form 
    \begin{equation}
      \cE(x)=\sum_i \Tr\cwich{M_i x}\sigma_i,
    \end{equation} 
    where $\{M_i\}$ is a POVM on $\mc{H}_A$ and $\sigma_i\in\S(\H_B)$ are states. Additionally, we assume that $\supp{\sigma}=\cup_i\supp{\sigma_i}=\H_B$, else we could redefine $\H_B$ as $\supp{\sigma}$. Then, 
    \begin{equation} 
        J_{A\rightarrow B} = \sum M_i\otimes\sigma_i \quad\text{and}\quad Q = \sum \swich{\rho\star M_i}\otimes \sigma_i.
    \end{equation} 
    Because the map in \cref{thm:JordanInverse} is CP, as seen in \cref{coro:JInverseCP}, $J_{A\leftarrow B}^{T_A}$ will be positive if $Q$ is positive, and $Q$ will be positive if every $\rho\star M_i$ is (with an if and only if when the $\sigma_i$ are orthogonal). This will happen if the channel is classical-quantum channels, that is when $\{M_i\}$ is a projective measurement, and $\rho$ is diagonal in a basis defined by this measurement.

    As a particular example of this last case, let $\rho=p\ketbra{0}+(1-p)\ketbra{1}$ and $J_{A\rightarrow B}=\ketbra{0,0}+\ketbra{1,+}$, the {\jamiol} matrix corresponding to the classical-quantum channel that keeps $\ketbra{0}$ constant and yields $\ketbra{+}$ on input $\ketbra{1}$. Then $\sigma=\fmig(\I+p\sigma_z+(1-p)\sigma_x)$ and the resulting state over time is 
    \begin{align}
        Q=\rho\star J_{A\rightarrow B}=\fmig\begin{bmatrix}
            2p & 0 & 0 & 0 \\
            0 & 0 & 0 & 0 \\
            0 & 0 & 1-p & 1-p \\
            0 & 0 & 1-p & 1-p 
        \end{bmatrix}. 
    \end{align} 
    If we write $J_{A\leftarrow B}$ we get 
    \begin{align}
        J_{A\leftarrow B}=\fmig\begin{bmatrix}
            1+p & p-1 & 0 & 0 \\
            p-1 & 1-p & 0 &  \\
            0 & 0 & 1-p & 1-p \\
            0 & 0 & 1-p & 1+p 
            \end{bmatrix}, 
    \end{align} 
    which is positive and left invariant under partial transposition. Note that $J_{A\leftarrow B}=\ketbra{0} \otimes\sigma_0+\ketbra{1} \otimes \sigma_1$, with 
    \begin{equation}
        \sigma_0=\fmig\begin{bmatrix}
            1+p & p-1 \\
            p-1 & 1-p \end{bmatrix},\quad  \sigma_1=\fmig\begin{bmatrix}
            1-p & 1-p \\
            1-p & 1+p \end{bmatrix}.
    \end{equation} 
    This {\jamiol} matrix corresponds to another measure-and-prepare map, this time with non-orthogonal measurements since $\Tr\cwich{\sigma_0\sigma_1}=p(1-p)$.
    \end{enumerate}
\end{example}

The example shows a type of channels where the symmetry of the set of states over time is broken: channels which reduce the coherence of the input states. In \cref{sec:UIcost} we discuss an example where this asymmetry is numerically shown in the cost function, rather than the set of couplings, acting as a proof that $K(\rho,\sigma)$ is not a symmetric function for symmetric cost matrices.

\subsection{Triangle inequality}
We will state a condition for the fulfilment of the triangle inequality.  For this purpose, we consider $\mc{Q}(\H_A:\H_B:\H_C)$ from \cref{def:QSHierarchy} and its dual, $\mc{Q}(\H_A:\H_B:\H_C)^*$. This dual provides us with a cone with respect to which we can define a partial order for cost matrices: 

\begin{thm}\label{thm:JPtrineq}
Consider Hilbert subspaces $\H_A$, $\H_B$ and $\H_C$. The inequality 
\begin{equation}
    \mc{K}_{AB}(\rho_A,\rho_B)+\mc{K}_{BC}(\rho_B,\rho_C)\geq \mc{K}_{AC}(\rho_A,\rho_C)
\end{equation} 
will be fulfilled for all $\rho_i\in\S(\H_i)$, $i\in\set{A,B,C}$ if the cost matrices fulfil the following identity: 
\begin{align}
     K_{AB}\otimes\I_C+\I_A\otimes K_{BC}-K_{AC}\otimes\I_B\in\mc{Q}(\H_A:\H_B:\H_C)^*.
\end{align}
\end{thm}

\begin{proof}
Consider first an admissible {\jamiol} matrix in systems $AB$ and a cost matrix $K_{AB}$. Because $\I_B=\Tr_C\cwich{J_{BC}}$ for any admissible {\jamiol} matrix in systems $BC$, and the partial associativity of the Jordan product \cite{horsman17,fitzsimons15} we can rewrite the cost as 
\begin{align}\begin{split}
    \Tr\cwich{K_{AB}(\rho\star J_{AB})} 
    &=\Tr\cwich{K_{AB}((\rho\star J_{AB})\star (\I_A\otimes\Tr_C\cwich{J_{BC}}))} \\
    &=\Tr\cwich{(K_{AB}\star (\rho\star J_{AB}))(\I_A\otimes\Tr_C\cwich{J_{BC}})} \\
    &= \Tr\cwich{((K_{AB}\star (\rho\star J_{AB}))\otimes\I_C)(\I_A\otimes J_{BC})} \\
    &=\Tr\cwich{(K_{AB}\otimes\I_C)((\rho\star (J_{AB}\otimes\I_C))\star (\I_A\otimes J_{BC}))} \\
    &=\Tr\cwich{\swich{K_{AB}\otimes\I_C}(\rho\star ((J_{AB}\otimes\I_C)\star (\I_A\otimes J_{BC})))}.
\end{split}\end{align} 
Similarly, because if a channel yields $\sigma$ as the image of $\rho$ its {\jamiol} matrix will fulfil $\Tr_A\cwich{\rho\star J_{AB}}=\sigma$ we can operate the cost for any admissible {\jamiol} matrices and cost $K_{BC}$ as: 
\begin{align}\begin{split}
    \Tr\cwich{K_{BC}(\sigma\star J_{BC})} 
    &=\Tr\cwich{K_{BC}((\Tr_A\cwich{\rho\star J_{AB}}\otimes\I_C)\star J_{BC})} \\
    &=\Tr\cwich{(J_{BC}\star K_{BC})(\Tr_A\cwich{\rho\star J_{AB}}\otimes\I_C)} \\
    &=\Tr\cwich{((\I_A\otimes J_{BC})\star (\I_A\otimes K_{BC}))(\rho\star J_{AB}\otimes\I_C)} \\
    &=\Tr\cwich{(\I_A\otimes K_{BC})((\rho\star (J_{AB}\otimes\I_C))\star (\I_A\otimes J_{BC}))} \\
    &=\Tr\cwich{(\I_A\otimes K_{BC})(\rho\star ((J_{AB}\otimes\I_C)\star (\I_A\otimes J_{BC})))}.
\end{split}\end{align} 
Finally, for systems $AC$, consider the link product \cite{chiribella09} of any two {\jamiol} matrices as the {\jamiol} matrix of $AC$: 
\begin{align}\begin{split}
    \Tr\cwich{K_{AC}(\rho\star J_{AC})}
    &=\Tr\cwich{K_{AC}(\rho\star \Tr_B\cwich{(J_{AB}\otimes\I_C)\star (\I_A\otimes J_{BC})})} \\
    &=\Tr\cwich{(K_{AC}\star \rho)\Tr_B\cwich{(J_{AB}\otimes\I_C)\star (\I_A\otimes J_{BC})}}\\
    &=\Tr\cwich{((K_{AC}\otimes\I_B)\star \rho)((J_{AB}\otimes\I_C)\star (\I_A\otimes J_{BC}))} \\
    &=\Tr\cwich{(K_{AC}\otimes\I_B)(\rho\star ((J_{AB}\otimes\I_C)\star (\I_A\otimes J_{BC})))}.
\end{split}\end{align} 

Let $K'=K_{AB}\otimes\I_C+\I_A\otimes K_{BC}-K_{AC}\otimes\I_B$. With these 3 equalities in hand, we can consider 3 optimal {\jamiol} matrices, indicated by the superindex ${}^o$, for the costs $\mc{K}_{AB}$ and $\mc{K}_{BC}$ and $\mc{K}_{AC}$. Then by using the previous expressions we can show that: 
\begin{align}\begin{split}
    \mc{K}_{AC} =&\Tr\cwich{K_{AC}(\rho\star J^o_{AC})}\leq\Tr\cwich{K_{AC}(\rho\star J_{AC})}\\
    =&\Tr\cwich{K_{AC}\otimes\I_B(\rho\star ((J^o_{AB}\otimes\I_C)\star (\I_A\otimes J^o_{BC})))} \\
    =&\Tr\cwich{(K_{AB}\otimes\I_C+\I_A\otimes K_{BC}-K')(\rho\star ((J^o_{AB}\otimes\I_C)\star (\I_A\otimes J^o_{BC})))} \\
    =&\Tr\cwich{(\I_A\otimes K_{BC})(\rho\star ((J^o_{AB}\otimes\I_C)\star (\I_A\otimes J^o_{BC})))} \\
    +&\Tr\cwich{\swich{K_{AB}\otimes\I_C}(\rho\star ((J^o_{AB}\otimes\I_C)\star (\I_A\otimes J^o_{BC})))} \\
    -&\Tr\cwich{K'\swich{\rho\star \swich{J_{AB}^o\star J_{BC}^o}}} \\
    =& \Tr\cwich{K_{AB}(\rho\star J^o_{AB})} +\Tr\cwich{K_{BC}(\sigma\star J^o_{BC})}-\Tr\cwich{K'\swich{\rho\star \swich{J_{AB}^o\star J_{BC}^o}}}\\
    =&\mc{K}_{AB}+\mc{K}_{BC}-\Tr\cwich{K'\swich{\rho\star \swich{J_{AB}^o\star J_{BC}^o}}}.
\end{split}\end{align} 
Finally, because $K'$ is in the dual of $\mc{Q}_3$, $\Tr\cwich{K'\swich{\rho\star \swich{J_{AB}^o\star J_{BC}^o}}}\geq0$ and therefore $\mc{K}_{AB}+\mc{K}_{BC}\geq \mc{K}_{AC}$.
\end{proof}

For the sake of completeness, this general statement can be converted to a statement about 
the triangle inequality for the cost $\mc{K}$ associated to a given cost matrix $K$:

\begin{coro}
Let $\H$ be a Hilbert space, $K\in\B(\H\otimes \H)$ and $\mc{K}$ the quantum optimal cost associated to $K$. 
Then $\mc{K}$ fulfils the triangle inequality if 
\begin{equation}
  K\otimes\I_3+\I_1\otimes K-K\otimes\I_2\in\mc{Q}(\H:\H:\H)^*,
\end{equation} 
where the subindices indicate different copies of the same Hilbert space $\H$.
\end{coro}

\subsection{Structural properties of the transport cost}

We would like to prove subadditivity of the cost for both inputs, unfortunately we can only see it for the second input. The following proposition shows this result as well as two consequences of the triangle inequality.

\begin{prop}\label{prop:subadditivity}
Let $p_x$ be a probability distribution. The optimal quantum cost fulfils the following:

\begin{enumerate}[i)]

    \item Convexity in the second argument: $\mc{K}\swich{\rho,\sum_x p_x\sigma_x} \leq\sum_x p_x\mc{K}\swich{\rho,\sigma_x}.$ \\
    
    Moreover, if  the triangle inequality is fulfilled:
    \item $\sum_x p_x\mc{K}(\rho_x,\sigma) \leq \mc{K}\swich{\sum_xp_x\rho_x,\sigma}+\sum_xp_x\mc{K}\swich{\rho_x,\sum_{x'}p_{x'}\rho_{x'}}$.\\
    
    \item  $\mc{K}\swich{\sum_x p_x\rho_x,\sigma}\leq\sum_x p_x \mc{K}\swich{\rho_x,\sigma} +\sum_{x'}p_{x'}\mc{K}(\sum_{x}p_{x}\rho_{x},\rho_{x'}).$
\end{enumerate}
\end{prop}

\begin{proof}
Consider optimal {\jamiol} matrices $J_o$ for $\swich{\rho,\sum_x p_x\sigma_x}$ and $J_x$ for $\swich{\rho,\sigma_x}$. Note that $J_{\Sigma}=\sum_x p_xJ_x$ is a {\jamiol} matrix with an associated channel that fulfils $\cE_{J_\Sigma}(\rho)=\sum_x p_x \sigma_x.$ Thus,
\begin{align}\begin{split}
    \mc{K}\swich{\rho,\sum_x p_x\sigma_x} &=\Tr\cwich{K\swich{\rho\star J_o}} \leq\Tr\cwich{K\swich{\rho\star J_{\Sigma}}} =\Tr\cwich{K\swich{\rho\star \swich{\sum_x p_xJ_x}}} \\
    &=\sum_x p_x\Tr\cwich{K\swich{\rho\star J_x}} 
    =\sum_x p_x \mc{K}(\rho,\sigma_x),
\end{split}\end{align} 
where we used the bilinearity of the Jordan product \cite{horsman17} and the linearity of the trace.\\

The second property is a direct consequence of the triangle inequality: 
\begin{align}\begin{split}
    \sum_x p_x\mc{K}\swich{\rho_x,\sigma} &\leq\sum_xp_x\swich{\mc{K}\swich{\rho_x,\sum_{x'}p_{x'}\rho_{x'}}+\mc{K}\swich{\sum_{x'}p_{x'}\rho_{x'},\sigma}} \\ &=\mc{K}\swich{\sum_xp_x\rho_x,\sigma}+\sum_xp_x\mc{K}\swich{\rho_x,\sum_{x'}p_{x'}\rho_{x'}}.
\end{split}\end{align}

Similarly, we can show the third property. Let $\rho=\sum_xp_x\rho_x$. Then, $\forall \;x$ \begin{align}\begin{split}
    p_x&\mc{K}(\rho,\sigma)\leq p_x\mc{K}(\rho,\rho_x)+p_x\mc{K}(\rho_x,\sigma) \\
    \Rightarrow  \quad \sum_xp_x&\mc{K}(\rho,\sigma)\leq \sum_xp_x\mc{K}(\rho,\rho_x)+\sum_xp_x\mc{K}(\rho_x,\sigma)\\
    \Rightarrow  \quad &\mc{K}(\rho,\sigma)\leq \sum_xp_x\mc{K}(\rho,\rho_x)+\sum_xp_x\mc{K}(\rho_x,\sigma),
\end{split}\end{align} 
where we first used the triangle inequality and then we added all the inequalities together. 
\end{proof}

\begin{remark}
A similar proof does not work for convexity on the first input and joint subadditivity because of the following. We will use subadditivity on the first input as an example. Let $J_o$ be the optimal {\jamiol} matrix for $\swich{\sum_xp_x\rho_x,\sigma}$. Starting on the left hand side we obtain 
\begin{align}
  \mc{K}{\swich{\sum_xp_x\rho_x,\sigma}}
     =\Tr\cwich{K\swich{\sum_xp_x\rho_x\star J_o}}
     =\sum_xp_x\Tr\cwich{K\swich{\rho_x\star J_o}}.
\end{align} 
At this point we can observe that the channel associated to $J_o$ does not necessarily have output $\sigma$ for each $\rho_x$ (unless $\sigma$ is pure) and we can not upper bound the associated cost with anything defined with the optimal channels for the pairs $(\rho_x,\sigma)$. In contrast, in the proof of \cref{prop:subadditivity} it was possible to define the joint channel $J_\Sigma$ because we could send $\rho$ to each element $\sigma_x$ of the ensemble $\{(p_x,\sigma_x)\}$ and that would in total define a channel that sends $\rho$ to $\sigma$. 

 Moreover, we can show joint convexity is false in general. Let $\sigma_x=\cE_{J_o}(\rho_x)$ and observe that letting $\sigma=\sum_xp_x\sigma_x$ in the previous equation we obtain:
\begin{equation}
    \mc{K}{\swich{\sum_xp_x\rho_x,\sum_xp_x\sigma_x}}\geq \sum_xp_x\mc{K}(\rho_x,\sigma_x).
\end{equation} 
This joint concavity is not general in the sense that we have the relation only for $\sigma_x=\cE_{J_o}(\rho_x)$, where the ensemble $\{(p_x,\rho_x)\}$ can be arbitrarily chosen, but the channel must be the one associated to the optimal {\jamiol} matrix for $\mc{K}(\rho,\sigma)$.  
\end{remark} 

For the next proposition, we consider bipartite initial and final spaces, and for simplicity we assume them to be the same: $\H_A=\H_B=\H_1\otimes\H_2$. In this setting we want to study how the cost for product states relates to related costs in $\H_1$ and $\H_2$ for different configurations of the joint cost matrix. 

\begin{prop}
\label{prop:tensorInequalities}
Let $\H_i$ be Hilbert spaces and $\rho_i,\sigma_i\in\S(\H_i)$ with $i=1,2$. Let $K_{12}$ be a cost matrix in $\B((\H_1\otimes\H_2)^{\otimes 2})$, and $K_i$ be cost matrices in $\B(\H^{\otimes 2}_i)$. Then the optimal transport cost of $\mc{K}\swich{\rho_1\otimes\rho_2,\sigma_1\otimes\sigma_2}$ fulfils the following: \begin{enumerate}[i)]
    \item If $K_{12}=K_1\otimes K_2$, then $\mc{K}(\rho_1\otimes\rho_2,\sigma_1\otimes\sigma_2)\leq \mc{K}(\rho_1,\sigma_1)\mc{K}(\rho_2,\sigma_2)$.
    \item If $K_{12}=K_1\otimes\I_2+\I_1\otimes K_2$, then $\mc{K}(\rho_1\otimes\rho_2,\sigma_1\otimes\sigma_2)\leq \mc{K}(\rho_1,\sigma_1)+\mc{K}(\rho_2,\sigma_2).$
\end{enumerate} 
\end{prop}

\begin{proof}
Objects in different subsystems commute and $J_1\otimes J_2$ is admissible for $\swich{\rho_1\otimes\rho_2,\sigma_1\otimes\sigma_2}$ if $J_i$ is admissible for $\rho_i,\sigma_i$, $i=1,2$. To show the first inequality, consider two optimal {\jamiol} matrices $J_1^o,J_2^o$ that optimise the costs between $\rho_i,\sigma_i$ with cost matrix $K_i$, $i=1,2$. Then, 
\begin{align}\begin{split}
    \mc{K}(\rho_1,\sigma_1)&\cdot \mc{K}(\rho_2,\sigma_2)=\Tr\cwich{K_1\swich{\rho_1\star J_1^o}}\cdot\Tr\cwich{K_2\swich{\rho_2\star J_2^o}}\\
    &=\Tr\cwich{\swich{K_1\swich{\rho_1\star J_1^o}}\otimes\swich{K_2\swich{\rho_2\star J_2^o}}} =\Tr\cwich{\swich{K_1\otimes K_2}\swich{\swich{\rho_1\otimes\rho_2}\star \swich{J_1^o\otimes J_2^o}}}\\
    &\geq \mc{K}(\rho_1\otimes\rho_2,\sigma_1\otimes\sigma_2).
\end{split}\end{align}

For the second inequality, consider the same {\jamiol} matrices as before. Then 
\begin{align}\begin{split}
  \mc{K}(\rho_1\otimes\rho_2,\sigma_1\otimes\sigma_2) 
    &=    \Tr\cwich{K_{12}\swich{\rho_1\otimes\rho_2}\star J_{12}^o} \\
    &\leq \Tr\cwich{\swich{K_1\otimes\I_2+\I_1\otimes K_2}\swich{\rho_1\otimes\rho_2}\star \swich{J_1^o\otimes J_2^o}} \\
    &=     \Tr\cwich{K_1 \rho_1\star J_1^o}\Tr\cwich{\rho_2\star J_2^o}+\Tr\cwich{\rho_1\star J_1^o}\Tr\cwich{K_2 \rho_2\star J_2^o} \\
    &=     \mc{K}(\rho_1,\sigma_1)+\mc{K}(\rho_2,\sigma_2).
\end{split}\end{align}
\end{proof}


\section{Unitary invariant cost} 
\label{sec:UIcost}
Let $\H$ be a finite dimensional Hilbert space of dimension $d$. Here we consider 
cost matrices which yield unitary invariant quantum optimal costs, that is 
\begin{align}
  \mc{K}(\rho,\sigma)=\mc{K}(U\rho U^*,U\sigma U^*) \quad \forall U\in U(d), \,\rho,\sigma\in\S(\H).
\end{align} 
This follows automatically if the cost matrix is invariant under simultaneous 
unitary transformations of both systems:
\begin{equation}
  K = (U\otimes U)K(U^*\otimes U^*) \quad \forall U\in U(d). 
\end{equation}

In this section, we consider the case  where the cost is unitarily invariant (UI), meaning it remains unchanged under a common unitary transformation of both the input and output states:
\begin{align}
  \kappa(\rho,\cE)=\kappa(U\rho U^*,\mathcal{U}\circ\cE\circ \mathcal{U}^{-1}) \quad \forall U\in U(d), \,\rho,\sigma\in\S(\H).
  \label{eq:def-UI}
\end{align} 
where we have defined the unitary channel $\mathcal{U}(X)=U X U^*$, and
we denote by $U(d)$ the group of unitary operators acting on the Hilbert space, of dimension $d$. We show the relation between the maps and states in the following commutative diagram:

\[ \begin{tikzcd}[
  column sep=5em,
  every arrow/.style={draw,mapsto}
]
\rho \arrow{r}{\cE} \arrow[swap]{d}{\mc{U}} & \sigma \arrow{d}{\mc{U}} \\%
U\rho U^* \arrow{r}{\mc{U}\circ\cE\circ\mc{U}^{-1}}& U\sigma U^*
\end{tikzcd}
\]

Throughout this section, we show that unitary invariance fully determines the cost matrix up to a single parameter. Later, we compute the cost for several relevant quantum channels, and derive a simplified optimisation problem to compute the cost $ \mc{K}(\rho,\sigma)$ valid for arbitrary states within the UI setting. In \cref{prop:commCost}, we obtain analytic expressions for the optimal cost in the special case of commuting states. \cref{sec:limitCommUI} explores the striking modification of the cost induced by embedding the problem into a higher-dimensional Hilbert space, and study its asymptotic behaviour. Finally, in \cref{example:UIcost}, we present an example that illustrates several noteworthy features. 

The following proposition shows that there is a single cost matrix (up to positive scaling) with the UI property:

\begin{prop}\label{prop:UIcost}
The only cost matrices that belong to the dual to the cone of states over time, 
assign cost $0$ to the identity channel according to \cref{prop:zerocosts} 
and commute with unitaries of the form $U\otimes U$ are positive multiples of 
\begin{equation}
  K_0=d\I-\S.
\end{equation}
\end{prop}

\begin{proof}
First, let us prove that if $J$ is the {\jamiol} matrix whose channel takes $\rho$ to $\sigma=\cE(\rho)$,
then the {\jamiol} matrix 
corresponding to the associated  channel
$\cE_U'=\mathcal{U}\circ\cE\circ \mathcal{U}^{-1}$ which maps $U\rho U^*$ to $U\sigma U^*$, is given
 by $J_U'=(U\otimes U)J (U^*\otimes U^*)$. 
 
For this purpose we start from the definition in \cref{eq:def-jamiol} 
\begin{align}
\begin{split}
J_U'&=\swich{\id\otimes\cE_U'}\swich{\S}=(\id\otimes\mathcal{U})\circ(\id\otimes\cE)\circ(\id\otimes\mathcal{U}^{-1})(\S) \\
    &= (\I\otimes U)\,(\id\otimes\cE)\!\left((\I\otimes U^*)\S(\I\otimes U)\right)(\I\otimes U)^*
\end{split}
\end{align}
Recall that the swap operator satisfies $\S(X\otimes Y)\S = Y\otimes X$, i.e.
    $(I\otimes V)\S = \S(V\otimes I)$. Using this relation twice we can write
  $(\I\otimes U^*)\S(\I\otimes U)=\S(U^*\otimes\I)(\I\otimes U)=\S(\I\otimes U)(U^*\otimes\I)=(U\otimes\I)\S(U^*\otimes\I)$, and thereby  pull the unitaries 
    out of the action of the map $\cE$, i.e.  
\begin{align}
\begin{split}
    J_{U}'
    &= (\I\otimes U)(U\otimes \I)\,(\id\otimes\cE)(\S)\,(U^{*}\otimes \I)(\I\otimes U^*) \\
    &= (U\otimes U)\,J\,(U^*\otimes U^*).
\end{split}
\end{align}
Now we can impose the UI of the cost function, $\kappa(\rho,\cE)=\kappa(U\rho U^*,\cE_U')$,
to obtain 
\begin{align}
\begin{split}
\Tr\cwich{K_0\swich{\rho\star J}}&=\Tr\cwich{K_0\swich{\swich{U\rho U^*}\star\swich{(U\otimes U)J (U^*\otimes U^*)}}}\\
&=\Tr\cwich{(U\otimes U)K_0(U^*\otimes U^*)(\rho\star J)}.
  \label{eq:equalU}
  \end{split}
\end{align} 
We require this equality to hold for all choices of input states and channels, i.e. for all elements of the set of states over time $Q=\rho\star J$. Since this set is spanning
 (\cref{prop:QPointedSpanning}), the scalar equality \cref{{eq:equalU}} translates into the operator equality $K_0=(U\otimes U)K_0(U^*\otimes U^*)$ or equivalently $\cwich{U\otimes U,K_0}=0$.
This must hold for all unitaries.

From the representation theory of $GL(d)$, and because the set of unitaries generates the whole of $GL(d)$ as an algebra, whose operations leave the commutator invariant, the only elements with this property are the symmetric and antisymmetric projectors. The vector space generated by these two projectors also has $\{\I,\S\}$ as a basis \cite{werner89}. Therefore $K_0=a(b\I-\S)$ with real $a$ and $b$, to preserve Hermiticity.

We can now impose the second condition, \cref{prop:zerocosts}: 
\begin{align}
    0=\Tr_B\cwich{\S\star K_0}=a\Tr_B\cwich{\S\star (b\I-\S)}=a\Tr_B\cwich{b\S-\I}=a(b-d)\I.
\end{align} 
We find that $b=d$. Finally, we see that the positivity of the cost requires $a$ to be positive: \begin{align}\begin{split}
    \Tr\cwich{K_0(\rho\star J)}&=a\Tr\cwich{(d\I-\S)(\rho\star J)}\\
    &=ad\Tr\cwich{\rho\star J}-a\Tr\cwich{\S(\rho\star J)}\\
    &=a\cwich{d-\Tr\cwich{\rho(\S\star J)}}.
\end{split}\end{align} The last term can be upper bounded by switching to the Choi matrix representation of the problem using \cref{lem:PTstrongDuality}. Recall that $\S^{T_A}=\ketbra{\Phi_+}$, where $\ket{\Phi_+}=\sum_i \ket{ii}$ is the unnormalised maximally entangled state and is such that $\Tr\cwich{\ketbra{\Phi_+}}=d$. Then \begin{align}
    \begin{split}
        \Tr\cwich{\rho(\S\star J)}=&\Tr\cwich{{\S^{T_A}}^{T_A}(\rho\star J)}=\Tr\cwich{\ketbra{\Phi_+}(\rho\star J)^{T_A}}=\expval{\rho^T\star J^{T_A}}{\Phi_+}\\
        =&\expval{\rho^T\star C}{\Phi_+}\leq d\Tr\cwich{\rho^T\star C}=d.
    \end{split}
\end{align} Thus we have that $a$ times a positive constant has to be positive, therefore $a$ is positive.
\end{proof}

We will also use the normalised version of $K_0$: $\tilde{K}_0=\I-\frac{1}{d}\S$, so that the maximum achievable cost is 1. With this specific cost matrix, there are some simple ways to write the cost associated to a channel depending on which representation of the channel we take. These forms will be useful later. 

\begin{remark} \label{remark:UICostForms}
Let $\rho$ be a state in a finite dimensional Hilbert space $\H$ and the cost matrix $\tilde{K}_0=\I-\frac{1}{d}\S$. Consider a channel $\cE$ with associated {\jamiol} and Choi matrices $J$, $C$, respectively and Kraus representation $\{\Kr_k\}$. Moreover, let $\rho=\sum_ip_i\ketbra{i}$ for some basis $\{\ket{i}\}$, $\ket{\rho}=\sum_i p_i\ket{ii}$ its vectorized form  and $\ket{\Phi_+}=\sum_i \ket{ii} $ is the unnormalised maximally entangled state. Then \begin{align}
\kappa(\rho,\mathcal{E})=\Tr\cwich{\tilde{K}_0(\rho\star J)}&=1-\frac{1}{d}\expval{\rho^T\star C}{\Phi_+}\label{eq:UICostForms1}\\
&=1-\frac{1}{d}\mathfrak{Re}(\bra{\rho} C \ket{\Phi_+})\label{eq:UICostForms1b}\\
&=1-\frac{1}{d}\sum_{ij}\frac{p_i+p_j}{2}\bra{i}\cE(\ketbra{i}{j})\ket{j}\label{eq:UICostForms2}\\
&=1-\frac{1}{d}\sum_{i}{p_i}\sum_{j}\mathfrak{Re}(\bra{i}\cE(\ketbra{i}{j})\ket{j})\label{eq:UICostForms_2b}\\
&=1-\frac{1}{d}\sum_k\mathfrak{Re}\swich{\Tr\cwich{\Kr_k^*}\Tr\cwich{\Kr_k \rho}}.\label{eq:UICostForms3}
\end{align}
\end{remark}

\begin{proof}
The term $1$ in every equation comes from the trace of the states over time with the identity, which is always one because the partial trace of a state over time is a state. The other part is associated to $\Tr\cwich{\S(\rho\star J)}$, and we will focus on that.

\cref{eq:UICostForms1} and \cref{eq:UICostForms1b} are a direct consequence of \cref{lem:PTstrongDuality}, recalling that $\S^{T_A}=\ketbra{\Phi_+}$.

\cref{eq:UICostForms2} comes from the definition of the {\jamiol} matrix, $J=\sum_{ij}\ketbra{i}{j}\otimes\cE(\ketbra{j}{i})$ and \cref{thm:JordanInverse}, which shows that in the product basis of the diagonal basis of $\rho$, $\rho\star J=\sum_{ij}\frac{p_i+p_j}{2}\ketbra{i}{j}\otimes\cE(\ketbra{j}{i})$. Then we define the swap operator in this product basis, $\S=\sum_{i'j'}\ketbra{i'}{j'}\otimes\ketbra{j'}{i'}$ and calculate $\Tr\cwich{S(\rho\star J)}$:

\begin{equation}\begin{split}
    \Tr\cwich{S(\rho\star J)}&=\Tr\cwich{\sum_{ij}\frac{p_i+p_j}{2}\sum_{i'j'}\swich{\ketbra{i'}{j'}\otimes\ketbra{j'}{i'}}\swich{\ketbra{i}{j}\otimes\cE(\ketbra{j}{i})}}\\
    &=\sum_{ij}\frac{p_i+p_j}{2}\sum_{i'j'}\delta_{ij'}\delta_{ji'}\bra{i'}\cE(\ketbra{j}{i})\ket{j'}=\sum_{ij}\frac{p_i+p_j}{2}\bra{i}\cE(\ketbra{i}{j})\ket{j}.
\end{split}\end{equation} 

Finally, for \cref{eq:UICostForms3} consider $J$ written as a function of the Kraus operators: \begin{equation}
J=(\id\otimes\cE)(\S)=\sum_k(\I\otimes \Kr_k)\S(\I\otimes \Kr_k^*)=\sum_k(\Kr_k^*\otimes \Kr_k)\S.
\end{equation} 
Then we add $\rho$ and the $\S$ from the cost:
\begin{equation}\begin{split}
\Tr&\cwich{\S(\rho\star J)}=\sum_k\Tr\cwich{\S\swich{\rho\star (\Kr_k^*\otimes \Kr_k)\S}}\\
&=\frac{1}{2}\sum_k\swich{\Tr\cwich{\S\swich{\rho(\Kr_k^*\otimes \Kr_k)\S}}+\Tr\cwich{\S\swich{(\Kr_k^*\otimes \Kr_k)\S\rho}}}\\
&= \frac{1}{2}\sum_k\swich{\Tr\cwich{\rho \Kr_k^*\otimes \Kr_k}+\Tr\cwich{\Kr_k\rho\otimes \Kr_k^*}}
=\frac{1}{2}\sum_k\swich{\Tr\cwich{\rho \Kr_k^*}\Tr\cwich{\Kr_k}+h.c.}\\
&=\sum_k\mathfrak{Re}\swich{\Tr\cwich{\Kr_k^*}\Tr\cwich{\rho \Kr_k}}.
\end{split}\end{equation}
This concludes the proof.
\end{proof}

\begin{remark}\label{remark:KrausCostInvariant}
It is well known that a single channel can have multiple Kraus representations \cite[Theorem 8.2]{nielsen10}. By \cref{eq:UICostForms3}, for every Kraus representation  of a channel 
\begin{equation}
    1-\frac{1}{d}\sum_k\mathfrak{Re}\swich{\Tr\cwich{\Kr_k^*}\Tr\cwich{\rho \Kr_k}}= \Tr\cwich{S(\rho\star J)},
\end{equation} and the {\jamiol} matrix is unique, therefore different Kraus representations of the same channel have the same associated cost. 

We can also show this explicitly. If two Kraus representations $\{E_i\}$, $\{F_j\}$ give rise to the same quantum channel, then there exists a unitary $U=(U_{ij})$ such that $E_i=\sum_j u_{ij}F_j$ \cite{nielsen10}. Then \begin{equation}
    \begin{split}
  \sum_i\mathfrak{Re}\swich{\Tr\cwich{E_i^*}\Tr\cwich{\rho E_i}}
  &=\sum_i\mathfrak{Re}\swich{\Tr\cwich{\sum_j \bar{U}_{ij}F_j^*}\Tr\cwich{\rho \sum_{j'} U_{ij'}F_{j'}}} \\
  &= \sum_{jj'}\swich{\sum_iU_{ij'}\bar{U}_{ij}}\mathfrak{Re}\swich{\Tr\cwich{ F_j^*}\Tr\cwich{\rho F_{j'}}} \\
  &=\sum_{jj'}\swich{U^T\swich{U^T}^*}_{j'j}\mathfrak{Re}\swich{\Tr\cwich{ F_j^*}\Tr\cwich{\rho F_{j'}}}\\
  &=\sum_{jj'}\delta_{jj'}\mathfrak{Re}\swich{\Tr\cwich{ F_j^*}\Tr\cwich{\rho F_{j'}}}=\sum_{j}\mathfrak{Re}\swich{\Tr\cwich{ F_j^*}\Tr\cwich{\rho F_{j}}}.
    \end{split}
\end{equation}
\end{remark}

In the following examples we calculate the cost associated to two important channels using the unitary invariant cost matrix. \cref{example:UIcost} in particular allows for some observations on the properties of the unitary invariant cost. 

\begin{example}[The replacement channel] 
Consider $\cE(x)=(\Tr x)\sigma$. Then the associated {\jamiol} matrix is $\I\otimes\sigma=\sigma_B$ and \begin{align}\begin{split}
            \kappa(\rho,\mathcal{E}_R)=\Tr\cwich{(d\I-\S)\rho_A\star \sigma_B}=&d-\Tr\cwich{\rho\sigma}\geq d-1,
        \end{split}\end{align} where the last  inequality can be seen, for example, using the trace and operator norms of $\rho$ and $\sigma$: $\Tr\cwich{\rho\sigma}\leq\norm{\rho}_{tr}\norm{\sigma}_{op}\leq 1.$\end{example}

         \begin{example}[Unitary channels]\label{example:UIcost} 
         Consider now the class of unitary channels: $\cE_U(x)=U x U^*$, where $U$ is a unitary operator. The associated {\jamiol} matrix is $J_{U}=U_B\S U_B^*$, where $U_B=\I\otimes U$. We can use \cref{eq:UICostForms3} in \cref{remark:UICostForms} to calculate the cost, since the Kraus operators associated to a Unitary channel are just $\{U\}$. This cost is then \begin{align}\begin{split}
            \kappa(\rho,\mathcal{E}_U)=&\Tr\cwich{(\rho_A\star (U_B\S U_B^*))(d\I-\S)}=d-\mathfrak{Re}\swich{\Tr\cwich{U^*}\Tr\cwich{\rho U}}
        \end{split}\end{align}

        We can find an expression for the optimal $U$ when $\rho$ and $\sigma$ are pure states. Due to the unitary invariance, w.l.o.g.~$\rho=\ketbra{0}$ and $\sigma=\ketbra{\varphi}$ where $\ket{\varphi}=\alpha\ket{0}+\sqrt{1-\alpha^2}\ket{1}$ with $\alpha\in\R_+$. The optimal unitary (in terms of maximising its trace) will leave $\left<\{\ket{0},\ket{1}\}\right>$ invariant and have 1 in the diagonal elements outside this subspace. Therefore the optimal (i.e. largest) value is \begin{align}
            \mathfrak{Re}\swich{\Tr\cwich{U^*}\Tr\cwich{\rho U}}=\alpha(d-2+2\alpha)
        \end{align} with associated cost \begin{equation}\label{eq:UIcost}
            \mathcal{K}_U(\rho,\sigma)=d-\alpha(d-2+2\alpha)=d(1-\alpha)+2\alpha(1-\alpha)=(1-\alpha)(d+2\alpha).
        \end{equation} Note that this optimum value is influenced by the action of the unitary on an invariant subspace orthogonal to the subspace where our state evolves; in particular it depends on its dimension. We will later address this further and show how we can remove this dependency in the limit when $d\rightarrow \infty$.

        If we now further restrict the problem to $d=2$, then this becomes $2(1-\alpha)(1+\alpha)=2(1-\alpha^2)=2(1-\abs{\braket{0}{\varphi}}^2)=2T(\ketbra{0},\ketbra{\varphi})^2$, where $T$ is the trace distance. Since it is the square of a distance, it cannot  be a distance, as it violates the triangle inequality.  

        In the limit of high $d\rightarrow \infty$ this quantity approximately becomes $d(1-\alpha)$, which, for small angles is $d(1-\cos\theta)\approx \frac{d}{2}\theta^2$, which is again the square of a distance.
 
    \end{example}

Importantly, the second example computes the optimal unitary cost, which can then be compared to the optimal cost over all channels. In \cref{fig:analyticalUnitaryInvariant}, we numerically observe that for dimension 4, the cost and the unitary cost are equal when the input states are pure but differ when the states are mixed. We analytically prove the case for general pure states in \cref{prop:analyticalUnitaryMaps}.

For the proof of \cref{prop:analyticalUnitaryMaps} we will first show \cref{lem:smolSupportSymmetry}, where we consider the case where the joint support of $\rho$ and $\sigma$, by which we mean $\H_S=\supp{\rho}+\supp{\sigma}$, is strictly contained in the overall Hilbert space $\H=\H_S\oplus\H_\bot$. This setting will appear again later in \cref{sec:limitCommUI}. 

\begin{figure}[ht]
  \centering
    \includegraphics[width=0.7\textwidth]{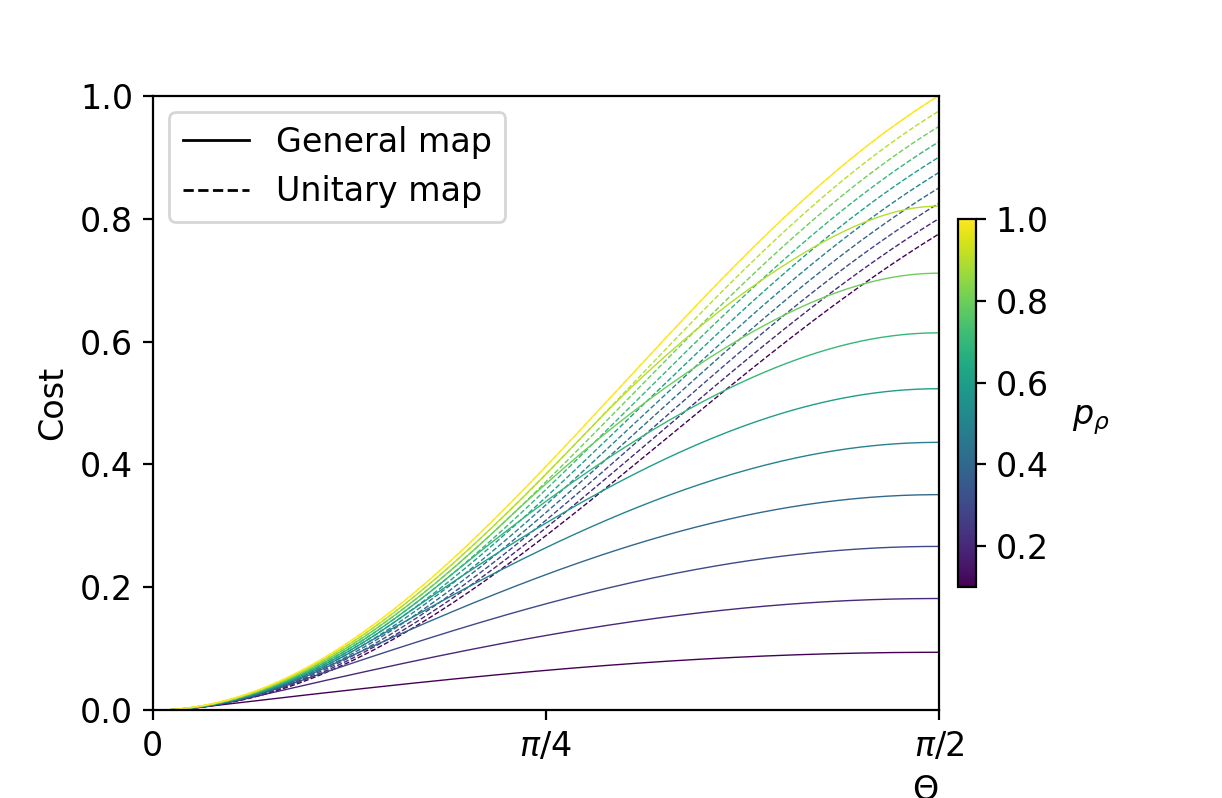} 
    \caption{Plot comparing the global optimal cost with the result of optimising only over unitaries 
             for various values of $p_\rho$, where the states are 
             $\rho=p_\rho\ketbra{0}+(1-p_\rho)\frac{1}{d}\I$
             and $\sigma=p_\rho\ketbra{\varphi}+(1-p_\rho)\frac{1}{d}\I$.} 
    \label{fig:analyticalUnitaryInvariant}
\end{figure}

\begin{lemma}\label{lem:smolSupportSymmetry}
Let $\rho$, $\sigma$ have joint support $\H_S\subseteq\H=\H_S\oplus\H_\bot$. Then there exists an optimal channel of the unitary invariant quantum optimal transport $\mc{K}(\rho,\sigma)$ such that its associated Kraus operators are of the form 
\begin{equation}
  \label{eq:smolSupportSymmetry}
  \Kr=\Kr_S\oplus c\Pi_\bot,
\end{equation} 
where $\Pi_\bot$ is the projector on the orthogonal, or embedding Hilbert space, $\H_\bot$. Therefore, the optimal channel acts as the identity channel on the embedding Hilbert space.
\end{lemma}

\begin{proof}
    We start by noting that our optimization problem \eqref{eq:SDPform} remains invariant under unitaries that act non-trivially only outside the joint support of the input and output states, i.e. they are of the form $(\Pi_S\oplus U_\bot)\otimes(\Pi_S\oplus U_\bot)$, where $\Pi_S$ is the projector onto $\H_S$ and $U_\bot$ are unitaries on $\H_\bot$ (see \cite{vallentin09} for a general discussion of SDP under symmetries).We show that the cost and constraints of the problem remain invariant under unitaries of this form: Given a valid {\jamiol} matrix $J$, we can construct a new {\jamiol} matrix $J_{U_\bot}=(\Pi_S\oplus U_\bot)\otimes(\Pi_S\oplus U_\bot)J(\Pi_S\oplus U_\bot^*)\otimes(\Pi_S\oplus U_\bot^*)$ that has the same cost and satisfies the same constraints as $J$ from \eqref{eq:SDPform}:
\begin{equation}\begin{split}
    \Tr\cwich{(\tilde{K}_0\star\rho)J_{U_\bot}}&=\Tr\cwich{(\tilde{K}_0\star\rho)((\Pi_S\oplus U_\bot)\otimes(\Pi_S\oplus U_\bot))J((\Pi_S\oplus U_\bot^*)\otimes(\Pi_S\oplus U_\bot^*))}\\
     &= \Tr\cwich{((\Pi_S\oplus U_\bot^*)\otimes(\Pi_S\oplus U_\bot^*))(\tilde{K}_0\star\rho)((\Pi_S\oplus U_\bot)\otimes(\Pi_S\oplus U_\bot))J} \\
     &=\Tr\cwich{(\tilde{K}_0\star\rho)J},
\end{split}\end{equation} 
where we have used the unitary invariance of the cost $\tilde{K}_0$ and the fact that $\rho$ and $\sigma$ commute with $U_\bot$. Similarly, the constraints of the problem are also invariant:
\begin{align}
    \Tr_A\cwich{\rho J_{U_\bot}}&=\Tr_A\cwich{\rho((\Pi_S\oplus U_\bot)\otimes(\Pi_S\oplus U_\bot))J((\Pi_S\oplus U^*_{\bot})\otimes(\Pi_S\oplus U_\bot^*))}\nonumber\\
    &=\Tr_A\cwich{((\Pi_S\oplus U_\bot^*)\otimes\I)\rho((\Pi_S\oplus U_\bot)\otimes(\Pi_S\oplus U_\bot))J(\I\otimes(\Pi_S\oplus U_\bot^*))} \\
    &=(\Pi_S\oplus U_\bot)\Tr_A\cwich{\rho J}(\Pi_S\oplus U_\bot^*)=(\Pi_S\oplus U_\bot)\sigma(\Pi_S\oplus U_\bot^*)=\sigma,\nonumber\\
    \nonumber \\
    \Tr_B\cwich{J_{U_\bot}}&= \Tr_B\cwich{((\Pi_S\oplus U_\bot)\otimes(\Pi_S\oplus U_\bot))J((\Pi_S\oplus U_\bot^*)\otimes(\Pi_S\oplus U_\bot^*))}\nonumber\\
    &=\Tr_B\cwich{(\I\otimes(\Pi_S\oplus U_\bot^*))((\Pi_S\oplus U_\bot)\otimes(\Pi_S\oplus U_\bot))J((\Pi_S\oplus U_\bot^*)\otimes\I)}\\
    &=(\Pi_S\oplus U_\bot)\Tr_B\cwich{J}(\Pi_S\oplus U_\bot^*)=(\Pi_S\oplus U_\bot)\I(\Pi_S\oplus U_\bot^*)=\I.\nonumber
    \end{align}
Due to the linearity of the cost and constraints, it follows that a twirled matrix $J'=\int dU_\bot J_{U_\bot}$ is also a valid {\jamiol} matrix, as it represents a convex combination of valid {\jamiol} matrices.  
That is, for any $J$ we can construct a twirled version $J'$ that retains the same cost. Hence, without loss of generality, we can optimize our cost over the set of {\jamiol} matrices satisfying the symmetry condition: $[(\Pi_S\oplus U_\bot)\otimes(\Pi_S\oplus U_\bot),J]=0$.

Since we wish to identify optimal Kraus operators we will work with the Choi matrix $C$: the canonical  Kraus representation can be readily obtained from $C=\sum_\Kr \ketbra{\Kr}{\Kr}$, where the (unnormalized) eigenstates are $\ket{\Kr}=\sum_{rs}\Kr_{rs}\ket{rs}$  are the vectorised form of the Kraus operator $\sum_{rs}\Kr_{rs}\ketbra{r}{s}$. In addition, we can use a basis for the partial transpose such that the Choi matrices inherit the symmetry of the problem. In particular, we consider a basis such that $\Pi_S$ is diagonal which will preserve the direct sum structure of $\Pi_S\oplus U_\bot$. In such a basis, the Choi matrix inherits the symmetry of the problem as follows: 
\begin{align}
    C'&=J'^{T_A}=\int dU_\bot J_{U_\bot}^{T_A}=\int dU_\bot((\Pi_S\oplus \bar{U}_\bot)\otimes(\Pi_S\oplus U_\bot))J^{T_A}((\Pi_S\oplus U_\bot^T)\otimes(\Pi_S\oplus U_\bot^*))\nonumber\\
   &= \int dU_\bot((\Pi_S\oplus \bar U_\bot)\otimes(\Pi_S\oplus  U_\bot))C((\Pi_S\oplus U_\bot^T)\otimes(\Pi_S\oplus U_\bot^*))\label{eq:twirledchoi}
\end{align}
where in the second equality we have used the properties of the partial trace given in Lemma \ref{lem:PTstrongDuality}. 
In order to ease the notation we will use Latin letters to label the $n$ elements of the basis of $\H_S$ and Greek letters to label the $d_\bot=d-n$  elements of the basis of $\H_\bot$.  We can expand the 16 terms appearing from the double direct sum in each side of the twirl in \cref{eq:twirledchoi}. All terms where each factor $(U_\bot)_{\alpha\beta}$ cannot be matched with a factor $(\bar U_\bot)_{\alpha \beta}$ will be zero. This follows from the invariance of the Haar measure $dU=dU'$ taking $U=W U' Z$ with $W= \sum_\alpha \mathrm{e}^{i\theta_\alpha} \ketbra{\alpha}{\alpha}$ and $Z= \sum_\alpha \mathrm{e}^{i\varphi_\alpha} \ketbra{\alpha}{\alpha}$.
Any unmatched factor $U_{\alpha\beta}$ picks up a phase $\mathrm{e}^{i(\theta_\alpha-\varphi_\beta)}$. Since the result must be independent of the value of these phases,  it must vanish––evident, for instance, by integrating over uniformly distributed phases in $[0,2\pi]$. For example, $\int dU_\bot (U_\bot)_{\alpha\beta}= \int dU_\bot (U_\bot)_{\alpha \beta} (U_\bot)_{\gamma \delta}=0$.

The only non-zero terms can be written as 
\begin{align}
    C_S&:=  (\Pi_S \otimes \Pi_S) C (\Pi_S \otimes \Pi_S)\label{eq:CSblock} \\
   V&:= \int dU_\bot (\Pi_S \otimes \Pi_S)C(U_\bot^T \otimes U_\bot^*)=\ket{v}\!\bra{\Phi_\bot} \mbox{ with }
    \ket{v}:=\sum_{ij,\beta} C_{\beta\beta; i j}^* \ket{ij} \\
   V^\dagger&:= \int dU_\bot(\bar U_\bot \otimes U_\bot)C(\Pi_S \otimes \Pi_S)=\ket{\Phi_\bot}\!\bra{v}   \\
   A&:= \int dU_\bot(\Pi_S \otimes U_\bot) C (\Pi_S \otimes U_\bot^*)= A_S \otimes  \Pi_\bot 
   \mbox{ where } A_S=\sum_{ij}(\sum_\alpha C_{i \alpha; j \alpha} )\ketbra{i}{j} \\
   B&:=\int dU_\bot(\bar U_\bot \otimes \Pi_S ) C (U_\bot^T \otimes \Pi_S )= \Pi_\bot\otimes B_S \mbox{ where } B_{S}=\sum_{ij}(\sum_\alpha C_{\alpha i;\alpha j})\ketbra{i}{j} \\
   D&:=  \int dU_\bot (\bar U_\bot \otimes U_\bot) C (U_\bot^T \otimes U_\bot^*) = a \Pi_\bot\otimes \Pi_\bot+ b \ketbra{\Phi_\bot}\label{eq:Dblock}.
\end{align}
where $\Pi_\bot=\sum \ketbra{\alpha}{\alpha}$ is the projector onto $\H_\bot$ and $\ket{\Phi_\bot}=\sum_{\alpha=1}^{d_\bot}\ket{\alpha \alpha}$ is the maximally entangled state in $\H_\bot\otimes\H_\bot$. In Eqs. (\ref{eq:CSblock}-\ref{eq:Dblock}) we used the $U(d)$ group integral \cite{puchala17}  $\int dU  U_{\alpha\beta}\bar{U}_{\nu\mu}=\frac{1}{d}\delta_{\alpha\nu}\delta_{\beta\mu}$ and the higher order contraction \begin{equation}\int dU U_{\alpha\beta}U_{\gamma\epsilon}\bar{U}_{\tau\xi}\bar{U}_{\nu\mu}=\frac{\delta_{\alpha\tau}\delta_{\gamma\nu}\delta_{\beta\xi}\delta_{\epsilon\mu}+\delta_{\alpha\nu}\delta_{\gamma\tau}\delta_{\beta\mu}\delta_{\epsilon\xi}}{d^2-1}-\frac{\delta_{\alpha\tau}\delta_{\gamma\nu}\delta_{\beta\mu}\delta_{\epsilon\xi}+\delta_{\alpha\nu}\delta_{\gamma\tau}\delta_{\beta\xi}\delta_{\epsilon\mu}}{d(d^2-1)}.\end{equation}

Notice that the symmetry under $U\otimes \bar U$ singles out the invariants $\ketbra{\Phi^+}{\Phi^+}$ and the identity, i.e. the so-called isotropic states \cite{werner89}. If we define the projector $\I_R=\Pi_\bot\otimes\Pi_\bot-\frac{1}{d_\bot}\ketbra{\Phi_\bot}$, we can further decompose the block $D$ in \eqref{eq:Dblock} in two diagonal blocks and write the twirled Choi matrix in the basis
\begin{equation} \label{eq:Cblock}
   C= 
\begin{pmatrix}
    \fbox{\(\begin{matrix} C_S &\hspace{1.5em}\ketbra{v}{\Phi_\bot}
 \\
 & \\ \ketbra{\Phi_\bot}{v} & a'\ketbra{\Phi_\bot}{\Phi_\bot} \end{matrix}\)} & \mathbf{0} \\
    \mathbf{0} & \begin{matrix} \fbox{$b'\I_R$} &  &  \\  & \fbox{$A$} &  \\  &  & \fbox{$B$} \end{matrix}
\end{pmatrix}
\end{equation}

Now we recall that a canonical set of Kraus operators can be obtained from the eigenstates of $C$. In particular, the eigenstates corresponding block $A=A'_S\otimes \Pi_\bot$ can be written as $\ket{\psi_S}\ket{\varphi_\bot}$ with $\ket{\psi_S} \in \H_S$
and $\ket{\varphi_\bot^*} \in \H_\bot$,  which correspond to  Kraus operators of the form $\Kr_A=\ketbra{\psi_S}{\varphi_\bot}$. 
Similarly the Kraus operator corresponding to the last block, $B$, are of the form $\Kr_B=\ketbra{\varphi_\bot'}{\psi_S'}$, those the central block $\Kr_R=\sum R_{\alpha\beta} \ketbra{\alpha}{\beta}$, and, finally, the first block's eigenstates must to be of the form $\ket{\Kr_o}=\sum_{ij} c_{ij}\ket{i j}+c \ket{\Phi_\bot}$ corresponding to a Kraus operator of the form $\Kr_o=\Kr_S\oplus c \Pi_\bot$.

To conclude the proof, we need to show that we can restrict to Kraus operators of the latter form.
We first note that $\Tr(\Kr_A)=\Tr(\Kr_B)=0$ 
and that $\Tr(\Kr_R \rho)=0$, and hence, $\Kr_A$, $\Kr_B$ and $\Kr_R$ do not contribute to the the non-trivial terms in the  cost, $\sum_\Kr\mathfrak{Re}\swich{\Tr\cwich{\Kr^*}\Tr\cwich{\Kr\rho }}$, as given in
\cref{eq:UICostForms3}. Moreover, these terms do not contribute to ensuring $\cE(\rho)=\sum_E \Kr \rho \Kr^*=\sigma$ as either the range or the image lies in the orthogonal subspace $\H_\bot$. Notably, since each term is positive semidefinite, any leakage from the support induced by one term cannot be canceled by another. Finally, the completeness relation $\sum_\Kr \Kr^*\Kr=\I$ can be satisfied solely with Kraus operators of the form $\Kr_o$ since $\Kr_o^* \Kr_o= \Kr_S^*\Kr_S \oplus c^2 \Pi_\bot$, without adversely affecting the cost.
Note that this amounts to a choice of Choi matrix that lives entirely in first block of \cref{eq:Cblock}, with $a'=1$ so that $\Tr_B C=\I$. 
\end{proof}

We next show that the optimization of the cost when the input and output are embedded in a larger Hilbert space can be written in terms of a quantum map acting only on the joint support of input and output.

\begin{thm}\label{thm:optEmbedd}
Let $\rho=\sum_i p_i \ketbra{i}$ and $\sigma$ have joint support $\H_S\subseteq\H=\H_S\oplus\H_\bot$;
with orthonormal basis $B_S=\set{\ket{i}}_{i=1}^n$ of $\H_S$ and $B_\bot=\set{\ket{\alpha}}_{\alpha=1}^{d_\bot}$ of $\H_\bot$, and $d=n+d_\bot$.
The optimal unitarily invariant cost is 
\begin{align}
    \mc{K}(\rho,\sigma)
      &= 1-\frac{1}{d}\max_{\{\Kr_k\}}\left(\mathfrak{Re}\swich{\sum_k\Tr\cwich{\Kr_k \rho}\Tr\cwich{\Kr_k^*}}+d_\bot\sqrt{\sum_k\abs{\Tr(\Kr_k \rho)}^2} \right)  \label{eq:Kembed3}\\
      &=1-\frac{1}{d}\max_{\mc{E}}\left(\mathfrak{Re}\swich{\sum_{ij}p_i\bra{i}\mc{E}(\ketbra{i}{j})\ket{j}}+d_\bot\sqrt{\sum_{ij}p_ip_j\bra{i}\mc{E}(\ketbra{i}{j})\ket{j}} \right),  \label{eq:Kembed2}
\end{align} 
where the maximisation is over CPTP maps $\mc{E}(\bullet)=\sum_k \Kr_k\bullet \Kr_k^*$ on 
$\B(\H_S)$ s.t. $\mc{E}(\rho)=\sigma$.
Equivalently, we can write
\begin{align}
   \label{eq:Kembed}
   \mc{K}(\rho,\sigma)
        =1-\frac{1}{d} &\max_{C_S}\left(\mathfrak{Re}(\bra{\rho\vphantom{\Phi_S}}C_S\ket{\Phi_S}) + d_\bot\sqrt{\bra{\rho}C_S\ket{\rho}} \right) \\
        &\phantom{===} 
         \text{s.t. } C_S\geq 0,\ \Tr_B C_S=\I_S,\ \Tr\cwich{\rho^T C_S}=\sigma, \nonumber
\end{align} 
where $\ket{\Phi_S}=\sum_{i=1}^n \ket{ii}$,
$\ket{\Phi_\bot}=\sum_{\alpha=1}^{d_\bot} \ket{\alpha \alpha}$, and the input is written in vectorized form $\ket{\rho}=\sum_{ii}p_{i}\ket{ii}$.
\end{thm}

\begin{proof}
From \cref{lem:smolSupportSymmetry}, we can take Kraus operators to be of the form $\Kr_k={\Kr_k}_S\oplus c_k\Pi_\bot$, such that $\set{{\Kr_k}_S}$ is a set of Kraus operators restricted to the support $\H_S$ and $c_k$ form a unit vector. The non-trivial part of the cost, starting from \cref{eq:UICostForms3}, then is
\begin{align}\label{eq:Boundem}
   \mathfrak{Re}\swich{\sum_k \Tr\cwich{\rho \Kr_k}\Tr\cwich{\Kr_k^*}}
= \mathfrak{Re}\swich{\sum_k \Tr\cwich{\rho {\Kr_k}_S}\Tr\cwich{{\Kr_k}_S^*}+d_\bot\sum_k c_k^*\Tr\cwich{\rho {\Kr_k}_S}}.
\end{align}  

We can remove the real part on the second term due to the phase freedom of each $c_k$ with respect to ${\Kr_k}_S$. This freedom will allow us to tune the phase of $c_k$ in each Kraus operator such that $c_k^*\Tr\cwich{\rho {\Kr_k}_S} = \abs{c_k}\abs{\Tr\cwich{\rho {\Kr_k}_S}}$, which is the maximum achievable real part. 

Now fix a set of Kraus operators on the support $\set{{\Kr_k}_S}$. Consider the vectors $\mathbf{u}=(c_k)$ and $\mathbf{v}=(\Tr\cwich{\rho {\Kr_k}_S})$. With this notation we are maximising the inner product between $\mathbf{v}$ and a unit vector $\mathbf{u}$. The Cauchy-Schwartz inequality states that $\abs{\braket{\mathbf{u}}{\mathbf{v}}}\leq \norm{\mathbf{u}}\norm{\mathbf{v}}=\norm{\mathbf{v}}$ and that the inequality is tight if and only if $\mathbf{v}$ and $\mathbf{u}$ are linearly dependent. Therefore the non-trivial part of the cost is 
\begin{equation}
    \mathfrak{Re}\swich{\sum_k \Tr\cwich{\rho {\Kr_k}_S}\Tr\cwich{{\Kr_k}_S^*}+d_\bot\sqrt{\sum_k \abs{\Tr\cwich{\rho {\Kr_k}_S}}^2}},
\end{equation} 
where $\sqrt{\sum_k\abs{\Tr\cwich{\rho \Kr_k}}^2}=\sqrt{\braket{\mathbf{v}}}$ is the norm of $\mathbf{v}$ and $c_k=\frac{\Tr\cwich{\rho {\Kr_k}_S}}{\sqrt{\sum_{k'} \abs{\Tr\cwich{\rho {\Kr_{k'}}_S}}^2}}$. Finally, we can take the maximum over all admissible channels to obtain the optimal channel. The  equivalent equations \cref{eq:Kembed2} and \cref{eq:Kembed} follow immediately from the general expressions \cref{eq:UICostForms_2b} and \cref{eq:UICostForms1b}, respectively.
\end{proof}

\begin{prop}
\label{prop:analyticalUnitaryMaps}
Consider two pure states in a $d$-dimensional Hilbert space $\H$, 
w.l.o.g. $\rho=\ketbra{0}$ and 
$\sigma=\ketbra{\psi}{\psi}$
where $\ket{\psi}=
\alpha\ket{0}+\sqrt{1-\alpha^2}\ket{1}$
with $\alpha\geq 0$. Let $\tilde{K}_0$ be the cost matrix. Then 
\begin{align}
    \mc{K}(\rho,\sigma)= (1-\alpha)(d+2\alpha)
\end{align} 
and the optimal channel is given by conjugation with the unitary 
\begin{equation}
U= \begin{bmatrix}
    \alpha & -\sqrt{1-\alpha^2} \\
    \sqrt{1-\alpha^2} & \alpha
\end{bmatrix}\oplus \I_{d-2}.
\end{equation}
\end{prop}

\begin{proof}
     Let $\set{\Kr_k}$ be an admissible set of Kraus operators for $\rho$ and $\sigma$. Because $\sigma$ is pure, $\Kr_k\rho \Kr_k^*=p_k\sigma$, with $p_k$ a probability distribution. From \cref{lem:smolSupportSymmetry} and \cref{thm:optEmbedd}, we can write $\Kr_k$ as follows 
     
     \begin{equation}
        \Kr_k=\sqrt{p_k}\left[
            \begin{tabular}{c | c}
                \begin{tabular}{c c}
                $\alpha$ & $\gamma_k$ \\
                $\beta$ & $\eta_k$
                \end{tabular} & 0 \\ \hline
                0 & $\Pi_\bot$ \\ 
                \end{tabular}\right],
    \end{equation} 
for some $\gamma_k,\eta_k\in\mathbb C$ and where we defined  $\beta=\sqrt{1-\alpha^2}\geq 0$.
    The trace-preserving condition $\sum_k \Kr_k^*\Kr_k=\mathbb I$ gives, from the lower-right entry of the $2\times 2$ block,
\begin{equation}
    \sum_k p_k\left(\abs{\gamma_k}^2+\abs{\eta_k}^2\right)=1,
    \label{eq:cond1}
\end{equation}
and, from the off-diagonal entry,
\begin{equation}
   \sum_k p_k\left(\alpha\gamma_k+\sqrt{1-\alpha^2}\eta_k\right)=0.
   \label{eq:cond2}
\end{equation}
Taking real parts, we define
\begin{equation}
       \gamma_R:=\sum_k p_k\,\mathfrak{Re}(\gamma_k),
    \qquad
      \eta_R:=\sum_k p_k\,\mathfrak{Re}(\eta_k).
      \label{eq:etaR}
\end{equation}
So that
\begin{equation}
    \alpha \gamma_R+\beta \eta_R=0.
    \label{eq:gm_eta_cond}
\end{equation}
For the cost we want to maximise, from \cref{thm:optEmbedd},
\begin{equation}\begin{split}
\mathfrak{Re}\swich{\sum_k\Tr\cwich{{\Kr_k}_S \rho}\Tr\cwich{{\Kr_k}_S^*}}+d_\bot\sqrt{\sum_k\abs{\Tr({\Kr_k}_S \rho)}^2} &=\mathfrak{Re}\swich{\sum_kp_k\alpha\swich{\alpha+\bar\eta_k}+d_\bot\sqrt{\sum_k p_k\alpha^2}}\\
&=f(\alpha,d_\bot)+\alpha\eta_R,
\end{split}\end{equation}  
where $f(\alpha,d_\bot)=\alpha(\alpha+ d_\bot)$ is a function of $\alpha$ and $d_\bot$ only, and therefore is a constant in the optimisation. Since $\alpha$ is positive, thus the optimisation reduces to maximising $\eta_R$ defined in \eqref{eq:etaR} subject to \eqref{eq:cond1} and \eqref{eq:cond2}.
Since $x\mapsto x^2$ is convex, Jensen's inequality yields
\begin{align}
    \gamma_R^2+\eta_R^2
    &=
    \left(\sum_k p_k\,\mathfrak{Re}(\gamma_k)\right)^2
    +
    \left(\sum_k p_k\,\mathfrak{Re}(\eta_k)\right)^2\nonumber\\
    &\le
    \sum_k p_k\left[(\mathfrak{Re}(\gamma_k))^2+(\mathfrak{Re}(\eta_k))^2\right]
    \le
    \sum_k p_k\left(\abs{\gamma_k}^2+\abs{\eta_k}^2\right)
    =
    1,\label{eq:smone}
\end{align}
where the last equality follows from the trace-preserving condition \eqref{eq:cond1}.
Combining this with \eqref{eq:cond2}, which for   for $\alpha>0$ reads,
\begin{equation}
    \gamma_R = -\frac{\beta}{\alpha}\eta_R,
\end{equation}
we get
\begin{equation}
    \left(1+\frac{\beta^2}{\alpha^2}\right)\eta_R^2=\frac{1}{\alpha^2}\eta_R^2 \le 1
    \;\Rightarrow\;
    \eta_R \le \alpha.
\end{equation}

For $\alpha=0$, \eqref{eq:cond2} gives $\eta_R=0$, so the bound holds in all cases.
Therefore,
\begin{equation}
    f(\alpha,d_\bot)+\alpha\eta_R \le f(\alpha,d_\bot)+\alpha^2=\alpha(2\alpha+d_{\bot})
\end{equation}
Hence,
\begin{equation}
\mc{K}(\rho,\sigma)=d-(f(\alpha,d_\bot)+\alpha\eta_R) \leq 
d-\alpha(2\alpha+d-2)=
(1-\alpha)(d+2\alpha)
\end{equation}
As seen in \cref{example:UIcost}, this bound is achieved by the unitary
\begin{equation}
    U=
    \begin{pmatrix}
        \alpha & -\beta \\
        \beta & \alpha
    \end{pmatrix}
    \oplus \Pi_\bot,
\end{equation}
for which $U\ket{0}=\ket{\psi}$ and $\eta_R=\alpha$, completing the proof.

\end{proof}

\subsection{Commuting density matrices}

In this section we continue the study of the unitary invariant cost problem by adding an extra constraint, that is that the states  commute: $[\rho,\sigma]=0$. First we show that the cost between arbitrary states can be bounded by the cost between the first states and the pinching of the second in the basis of the first, which will yield commuting states. Then we see that the optimal quantum transport cost between commuting states can be analytically calculated and that the optimal map is purely quantum. We also bound the cost between commuting states given by classical maps and show it is much larger than the general quantum cost in general.

\begin{prop}\label{prop:diagonalLowerBound}
Let $\H$ be a finite dimensional Hilbert space, $\rho,\,\sigma\in S(\H)$ with $\rho$ diagonal in the basis $\{\ket{i}\}$ and $\cE_\rho$ the pinching map in this basis, $\cE_{\rho}(x) =\sum_i\expval{x}{i}\ketbra{i}$. Let $\mc{K}$ be the quantum optimal transport cost associated to the unitary invariant cost matrix $K_0$. Then 
\begin{equation}
\mc{K}(\rho,\sigma)\geq\mc{K}(\rho,\cE_\rho(\sigma)).
\end{equation}
\end{prop}

\begin{proof}
Let $\rho=\sum_ip_i\ketbra{i}$ and consider the Choi matrix $C$ associated to a channel $\cE$ such that $\cE(\rho)=\sigma$. This matrix will be $C=\swich{\sum_{ij}\ketbra{i}{j}\otimes\cE(\ketbra{i}{j})}$. The diagonal elements of this matrix, which need to be positive, are $\expval{\cE(\ketbra{i})}{j}$ and fulfil  
\begin{equation}
  \sum_{j}\expval{\cE(\ketbra{i})}{j}=\Tr\cwich{\I\cE(\ketbra{i})}=1. 
\end{equation} 
Thus these form a classical stochastic map $p(j|i)=\expval{\cE(\ketbra{i})}{j}$. Also note that because $\cE(\rho)=\sigma$ the Choi matrix must fulfil $\sigma=\Tr_A\cwich{\rho C}= \sum_{i}p_i\cE(\ketbra{i})$. We can apply the pinching $\cE_{\rho}$ to this equation to obtain \begin{equation}
    \label{eq:pinchsigmaValue}
  \cE_{\rho}(\sigma) = \sum_{i}p_i\cE_{\rho}\swich{\cE(\ketbra{i})} 
                     = \sum_{ij}p_i\expval{\cE(\ketbra{i})}{j}\ketbra{j} 
                     = \sum_{ij}p_ip(j|i)\ketbra{j},
\end{equation} which will be useful later. 

We can now bound the cost associated to each channel with an expression of the associated classical stochastic map. Note that $\bra{i}\cE(\ketbra{i}{j})\ket{j}$ are the off diagonal elements of $C$ that complete a $2\times2$ minor with $\expval{\cE(\ketbra{i})}{i}$ and $\expval{\cE(\ketbra{j})}{j}$. Therefore,
\begin{equation} \label{eq:choiCoeffBound}
   \mathfrak{Re}\swich{ \bra{i}\cE(\ketbra{i}{j})\ket{j}}\leq \sqrt{\expval{\cE(\ketbra{i})}{i}\expval{\cE(\ketbra{j})}{j}}=\sqrt{p(i|i)p(j|j)}.
\end{equation} 
Finally, with \cref{eq:choiCoeffBound} we can calculate the bound on the non-trivial part of the cost associated to an admissible Choi matrix. Let $\ket{\Phi_+}=\sum_i{\ket{ii}}$ again be the unnormalised maximally mixed state, then: 
\begin{equation}
    \label{eq:pinchingBound}
    \begin{split}
\Tr&\cwich{\ketbra{\Phi_+}(\rho\star C)}
=\sum_{i'j'ij}\bra{i'i'} \frac{p_i+p_j}{2}\swich{\ketbra{i}{j}\otimes\cE(\ketbra{i}{j})}\ket{j'j'} \\
&=\sum_{ij}\frac{p_i+p_j}{2}\bra{i}\cE(\ketbra{i}{j})\ket{j}
=\sum_{i}p_i\expval{\cE(\ketbra{i})}{i}+\sum_{i\neq j}\frac{p_i+p_j}{2}\bra{i}\cE(\ketbra{i}{j})\ket{j} \\
& = \sum_{i}p_i\expval{\cE(\ketbra{i})}{i}+\sum_{i < j}\swich{p_i+p_j}\mathfrak{Re}\swich{\bra{i}\cE(\ketbra{i}{j})\ket{j}} \\
&\leq\sum_{i}p_ip(i|i)+\sum_{i< j}\swich{p_i+p_j}\sqrt{p(i|i)p(j|j)} 
= \sum_{i}p_ip(i|i)+\sum_{i\neq j}p_i\sqrt{p(i|i)p(j|j)}.
\end{split}\end{equation} 

Fix an admissible channel between $\rho$ and $\sigma$ and its associated classical stochastic map $p(j|i)$. Let $C_p = \ketbra{\phi}+\sum_{i\neq j}p(j|i)\ketbra{ij}$ with $\ket{\phi}=\sum_i\sqrt{p(i|i)}\ket{ii}$. This is clearly positive and \begin{align}
\Tr_B\cwich{C_p}&=\Tr_B\cwich{\ketbra{\phi}+\sum_{i\neq j}p(j|i)\ketbra{ij}}=\sum_ip(i|i)\ketbra{i}+\sum_{i\neq j}p(j|i)\ketbra{i}=\I,\\
\Tr_A\cwich{\rho^T C_p}&=\Tr_A\cwich{\sum_kp_k(\ketbra{k}\otimes\I)\swich{\ketbra{\phi}+\sum_{i\neq j}p(j|i)\ketbra{ij}}}\\
&=\sum_ip_ip(i|i)\ketbra{i}+\sum_{i\neq j}p_ip(j|i)\ketbra{j}=\sum_{ij}p(j|i)p_i\ketbra{j}=\cE_\rho(\sigma),
\end{align} where the last equality was seen in \cref{eq:pinchsigmaValue}. Therefore $C_p$ is a Choi matrix with associated channel such that $\cE_{C_p}(\rho)=\cE_\rho(\sigma)$. The elements $\bra{i}\cE_{C_p}(\ketbra{i}{j})\ket{j}$ are \begin{align}
    \bra{i}\cE_{C_p}(\ketbra{i}{j})\ket{j}&=\bra{i}\Tr_A\cwich{(\ketbra{j}{i}\otimes \I)\swich{\ketbra{\phi}+\sum_{i'\neq j'}p(j'|i')\ketbra{i'j'}}}\ket{j}\\
    &=\bra{i}\Tr_A\cwich{(\ketbra{j}{i}\otimes \I)\ketbra{\phi}}\ket{j}=\sqrt{p(i|i)p(j|j)},
\end{align} which is the tight version of \cref{eq:choiCoeffBound}. This means the bound \eqref{eq:pinchingBound} can be made tight for every admissible classical stochastic map between $\rho$ and $\sigma$ in the problem between $\rho$ and $\cE_\rho(\sigma)$ by choosing the adequate channel $C_p$. In particular, we can tighten this bound in the problem between $\rho$ and $\cE_\rho(\sigma)$ for a classical stochastic map associated to an optimal channel between $\rho$ and $\sigma$, thus yielding \begin{align}
    \mc{K}(\rho,\sigma)&\geq d-\sum_{i}p_ip(i|i)+\sum_{i\neq j}p_i\sqrt{p(i|i)p(j|j)}=\Tr\cwich{K_0(\rho\star C_p^{T_A})}\geq\mc{K}(\rho,\cE_\rho(\sigma)),
\end{align} finalising the proof.
\end{proof}

\begin{prop}
    \label{prop:commCost}
Let $\rho$ and $\sigma$ commute. In a common diagonal basis they can be written as $\rho=\sum_ip_i\ketbra{i}$, $\sigma=\sum_iq_i\ketbra{i}$. Let $\tilde{K}_0$ be the cost matrix. Then \begin{align}
    \label{eq:costCommStates}
    \mc{K}(\rho,\sigma)= \frac{1}{d}\swich{d-\sum_{i j} p_i \sqrt{\min\{1,\tfrac{q_i}{p_i}\}\min\{1,\tfrac{q_j}{p_j}\}}}
\end{align}
\end{prop}

\begin{proof}

We have seen in the proof of \cref{prop:diagonalLowerBound} that for every admissible channel there is an associated stochastic map and that for each stochastic map with an associated channel there is a channel that makes \cref{eq:pinchingBound} tight. Therefore the problem is equivalent to the following optimisation over classical stochastic maps: \begin{align}
\min_{p(j|i)} d-\sum_ip_ip(i|i)-\sum_{i\neq j}p_i\sqrt{p(i|i)p(j|j)},
\end{align} such that $q_j=\sum_i p(j|i)p_i$. Note that only the diagonal terms of the classical stochastic map contribute to the cost and we want to maximise them. This is equivalent to the well known problem of writing the total variation distance as a classical optimal transport problem (see \cite[Section~6]{villani08}, also \cite[Example~4.14]{vanHandel16} for a detailed explanation of the total variation distance as a classical transport problem). We succinctly give the solution to this problem for completion. 

The maximum value of $0\leq p(i|i)\leq 1$ subject to $q_i=\sum_j p(i|j)p_j\geq p(i|i)p_i$, is $p(i|i)=\frac{q_i}{p_i}\leq1$ if $p_i\geq q_i$ and $p(i|i)=1$ if $p_i<q_i$, or more succinctly \begin{equation}\label{eq:optimalPii}
    p(i|i)=\min\{1,\frac{q_i}{p_i}\}.
\end{equation} This will maximise the amount of weight the map leaves in place. If $p_i=0$ then $q_i/p_i$ is not well defined. In this case we take the maximum value $p(i|i)=1$, as this value appears in the cost and we want to have as large as possible for the optimisation, but it does not appear in the constraint (it is multiplied by $p_i=0$). 

To finalise we need to find the values of the terms $p(j|i)$ for $i\neq j$. If $p(i|i)=1$, then $p(j|i)=0$ for $i\neq j$, as the sum of these positive terms is 1. When $p(i|i)<1$, given that the terms $p(j|i)$, $i\neq j$, do not contribute to the cost, the optimal transport plan $p(j|i)$ can take a variety of values as long as the admissibility of the overall coupling, given by the condition $q_j=\sum_i p(j|i)p_i$, is preserved.
\end{proof}

\begin{remark}
    The optimal channel associated to the unitary invariant quantum optimal transport problem between commuting states with common basis $\set{\ket{i}}$ is, as we have seen, in Choi matrix form: 
    \begin{equation}
    C= \ketbra{\phi}+ \sum_{i\neq j}p(j|i)\ketbra{ij},
    \end{equation} 
    with $\ket{\phi}=\sum_i \sqrt{p(i|i)}\ket{ii}$ and $p(i|i)$ as previously defined in \cref{eq:optimalPii}, which has rank at most $d^2-d+1$. 
    
    We can further study the structure of these maps by looking at their Kraus matrices. We have the unnormalised eigenvectors of the Choi matrix: $\{\ket{\Kr_k}\}=\{\sum_{i}\sqrt{p(i|i)}\ket{ii};\allowbreak \sqrt{p(j|i)}\ket{ij}, \; i\neq j\}$, such that $C=\sum_k\ketbra{\Kr_k}$. The associated Kraus matrices are the un-vectorised elements: \begin{equation}\label{eq:commKraus}
    \{\Kr_k\}= \set{\sum_{i}\sqrt{p(i|i)}\ketbra{i}; \sqrt{p(j|i)}\ketbra{j}{i}, \; i\neq j}.
    \end{equation} 

    These Kraus matrices have the characteristic of not being able to generate coherence, but, in the case of $\sum_{i}\sqrt{p(i|i)}\ketbra{i}$, not completely destroy it. In the context of the theory of quantum coherence as a resource, these matrices are incoherent operations (IO), but not strictly incoherent (SIO) \cite{streltsov17}. 
\end{remark} 

\begin{remark}
If $\rho$ and $\sigma$ commute, we can consider the case where we restrict our channels to classical channels to see how it relates to known classical distances and whether classical maps are optimal in the quantum setting. A quantum map will be classical if its {\jamiol} matrix is diagonal in a product basis. W.l.o.g.~$\rho$ and $\sigma$ are diagonal in the computational basis and the {\jamiol} matrix be diagonal in the product of canonical basis. As seen in  \cref{remark:classicalBayes} a state over time in this case will be of the form $Q=\sum_{ij}p_{ij}\ketbra{ij}$, with $p_{ij}$ a joint probability distribution, that is a classical coupling. It is immediate to see that the associated cost with cost matrix $\tilde{K}_0=\I-\frac{1}{d}\S$ is related to the total variation distance as follows: $\Tr\cwich{(\I-\frac{1}{d}\S)Q}=\sum_{ij}p_{ij}- \frac{1}{d}\sum_ip_{ii}=1-\frac{1}{d}\sum_ip_{ii}=1-\frac{1}{d}+\frac{1}{d}\frac{1}{2}|\rho-\sigma|\geq 1-\frac{1}{d}$.\footnote{We slightly abuse notation here be denoting the classical probability distribution associated to the diagonal of states $\rho,\sigma$ in the canonical basis by $\rho,\sigma$.} We obtain the cost in \cref{prop:commCost} without the term $-\frac{1}{d}\sum_{i\neq j}p_i\sqrt{p(i|i)p(j|j)}$. The fact that this cost is larger than $1-\frac{1}{d}$ shows that a large gap can exist between the cost associated to classical channels and the optimal quantum cost, which can go to zero by definition.
\end{remark}


\subsection{Limit $d\rightarrow \infty$}
\label{sec:limitCommUI}

With the analytical formula of the particular case of commuting states for the normalised cost matrix $\tilde{K}_0=\I-\frac{1}{d}\S$, we can consider what happens if we embed our finite dimensional states into a larger $d$ dimension system and then take the limit $d\rightarrow \infty$. First, let us rewrite \cref{eq:costCommStates}. We can split the sum into a sum over $i$ and a sum over $j$ as 
\begin{equation}
    \frac{1}{d}\swich{d-\sum_ip_ip(i|i)-\sum_{i\neq j}p_i\sqrt{p(i|i)p(j|j)}}=\frac{1}{d}\swich{d-\swich{\sum_ip_i\sqrt{p(i|i)}}\swich{\sum_{j}\sqrt{p(j|j)}}}.
\end{equation} 
We should address what happens to $p(i|i)=\min\{1,\frac{q_i}{p_i}\}$ when $p_i=0$. Because $p_i=0$, $p(i|i)$ does not affect the outcome of applying the map to the relevant state. As the value of $p(i|i)$ does not matter, we take the choice $p(i|i)=1$. 

With this we can calculate the limit. Let $\rho,\sigma$ be commuting states in a finite dimensional Hilbert space $\H_S$ of dimension $n$, such that the joint support of $\rho$, $\sigma$ is $\H_S$. For a dimension $d\geq n$ we have a finite dimensional Hilbert space and a natural embedding that allows us to consider $\rho$ and $\sigma$ in this space. We can take a basis in which $\rho$ and $\sigma$ are diagonal and compute the cost. $\sum_ip_i\sqrt{p(i|i)}$ is fixed regardless of dimension. $\sum_{j}\sqrt{p(j|j)}$  has a fixed part, the sum of $p(j|j)$ in the support of $\rho$, and $d-n$ times $p(j|j)=1$ for the part not in the support, which add up to $d-n$. We call these fixed parts $N$ and $M$, respectively, and calculate the cost: \begin{equation}\label{eq:commStatesLimit}\begin{split}
\mc{K}_d(\rho,\sigma)&=\frac{1}{d}\swich{d-N\swich{M+d-n}}=1-\frac{NM}{d}-N+\frac{Nn}{d}\\
&\xrightarrow[d \to \infty]{} 1-N
=1-\sum_ip_i\sqrt{p(i|i)}.
\end{split}\end{equation} 
We can further develop this expression to write it as a function of $p_i$ and $q_i$, the diagonal elements of $\rho$ and $\sigma$, only: \begin{equation}
    \begin{split}
        \mc{K}_\infty(\rho,\sigma)
        & =1-\sum_ip_i\sqrt{p(i|i)}
        =\sum_ip_i(1-\sqrt{p(i|i)})
        =\sum_i\sqrt{p_i}(\sqrt{p_i}-\sqrt{p(i|i)p_i})\\
        & =\sum_{q_i<p_i}\sqrt{p_i}(\sqrt{p_i}-\sqrt{q_i}),
    \end{split} 
\end{equation} 
where the last equality comes from the definition of the optimal $p(i|i)=\min\{1,\frac{q_i}{p_i}\}$ in \cref{eq:costCommStates}. 

Using \cref{thm:optEmbedd} the general case is immediate:

\begin{thm}
Let $\rho,\sigma$ be states in a finite dimensional Hilbert space $\H_S$ of dimension $n$, such that the joint support of $\rho$, $\sigma$ is $\H_S$. Let $d\geq n$ and $\H_d=\H_S\oplus\H_\bot$ 
be a finite dimensional Hilbert space of dimension $d$. In $\H_d$, consider the cost matrix
$\tilde{K}_d=\I_d-\frac{1}{d}\S_d$. We denote the cost associated to the embedded 
$\rho,\, \sigma$ in a larger Hilbert space with cost matrix $\tilde{K}_d$ as $\mc{K}_d(\rho,\sigma)$. Then,
\begin{equation}
    \label{eq:genralLimit}
    \mc{K}_\infty(\rho,\sigma)=\lim_{d\to\infty}\mc{K}_d(\rho,\sigma)= 1-\max_{\set{\Kr_k}}\sqrt{\sum_k\abs{\Tr\cwich{\rho \Kr_k}}^2}=1-\max_{\mathcal E:\,\mathcal E(\rho)=\sigma}
F(\rho,\mathcal E),
\end{equation}
where the maximisation is over all admissible Kraus representations 
$\{E_k\}$ of channels $\mathcal E$ acting on $\H_S$, and where
$
F(\rho,\mathcal E)
:=\sqrt{\sum_k\abs{\Tr(\rho E_k)}^2}
$
denotes the entanglement fidelity (also called channel fidelity) \cite{schumacher96,watrous18} of the state $\rho$ under the channel $\mathcal E$.
\end{thm}

 The entanglement fidelity measures how well the channel $\cE$ preserves not only the state $\rho$, but also its entanglement with a reference system.

Before we give the proof, note that we now have two optimisation problems in the support Hilbert space $\H_S$, the original problem, with cost function equivalent to\begin{equation}
    \sum_k\Tr\cwich{E_k\rho}\Tr\cwich{E_k^*}
\end{equation} and the limit problem, with cost function \begin{equation}
    \sqrt{\sum_k\abs{\Tr\cwich{\rho \Kr_k}}^2}.
\end{equation} Both are optimised over the same set of quantum channel. As the costs are non-equivalent, these two problems are different they in general should have different optimal costs and channels. That said, numerical evidence suggests that these channels are equal or at least very close. 


\begin{proof}
Using \cref{eq:Kembed} in \cref{thm:optEmbedd} we can immediately obtain the result: 
\begin{equation}\begin{split}
    \mc{K}_\infty(\rho,\sigma) 
    &= \lim_{d\to\infty}1-\frac{1}{d}\max_{\{\Kr_k\}}\swich{\mathfrak{Re}\swich{\sum_k\Tr\cwich{\Kr_k \rho}\Tr\cwich{\Kr_k^*}}+(d-n)\sqrt{\sum_k\abs{\Tr\cwich{\rho \Kr_k}}^2} } \\
    &=1-\max_{\{\Kr_k\}}\sqrt{\sum_k\abs{\Tr\cwich{\rho \Kr_k}}^2},
\end{split}\end{equation}
as claimed. 
\end{proof}

As the set of Kraus operators that optimises \cref{eq:genralLimit} is \emph{a priori} different to the optimal set of Kraus operators on any finite dimension, it is unclear if $\mc{K}_\infty(\rho,\sigma)$ is efficiently computable. In the following Remark we see that we can efficiently compute $\mc{K}_\infty(\rho,\sigma)$, as well as the optimal channel, through an SDP on the joint support of $\rho,\sigma$.

\begin{remark}
    The cost in the limit, $\mc{K}_{\infty}(\rho,\sigma)$, can be efficiently calculated from the SDP \begin{equation}
       \begin{split}
    \max_J &\quad \quad\; \Tr\cwich{\rho\S\rho J} \\
    \text{s.t.} & \quad\:\left\{\begin{aligned}
        &\Tr_A\cwich{\rho J}=\sigma\\
        &\Tr_B J=\I\\
        &J^{T_A}\geq0
       \end{aligned} \right.,
\end{split}
    \end{equation} where $S,J\in\B(\H_S\otimes\H_S)$, $\H_S=\supp\rho\cup\supp\sigma$, and we simplified $(\rho\otimes \I)$ to $\rho$.
\end{remark}

Before the proof, note that the constraints are the same as in \cref{eq:SDPform}, they denote that $J$ is a {\jamiol} operator such that $\cE_J(\rho)=\sigma$.

\begin{proof}
 We only need to write the objective function in \cref{eq:genralLimit} as a function of the {\jamiol} matrix instead of the Kraus operators. Recall that $J=\sum_k(\I\otimes \Kr_k)\S(\I\otimes \Kr_k^*)$ and that the optimisation in \cref{eq:genralLimit} is over Kraus operators $E_k:\H_S\rightarrow\H_S$. Let $\S,J\in\B(\H_S\otimes\H_S)$, then: \begin{equation}\begin{split}
  \Tr\cwich{\rho\S\rho J}
    &= \Tr\cwich{\rho\S\rho \sum_k(\I\otimes \Kr_k)\S(\I\otimes \Kr_k^*)}
     = \sum_k \Tr\cwich{\S(\rho\otimes \rho)(\Kr^*_k\otimes \Kr_k)\S} \\
    &= \sum_k \Tr\cwich{\rho \Kr_k^*}\Tr\cwich{\rho \Kr_k}
     = \sum_k \abs{\Tr\cwich{\rho \Kr_k}}^2.
\end{split}\end{equation} 
Because the square root is monotonic, maximising this quantity is equivalent to 
maximising the square root, allowing us to efficiently compute $\mc{K}_\infty(\rho,\sigma)$ as 
\begin{equation}
    \mc{K}_\infty(\rho,\sigma) = 1-\sqrt{\max_J\Tr\cwich{\rho\S\rho J}},
\end{equation} where the optimisation is over {\jamiol} matrices of channels mapping $\rho$ to $\sigma$. 
\end{proof}

\begin{remark}
We can see that the general formula for the limit reduces to the commuting case correctly. 
If $[\rho, \sigma]=0$, the Kraus operators for the optimal channel are of the form 
$\{\Kr_k\}= \left\{\sum_{i}\sqrt{p(i|i)}\ketbra{i}; \sqrt{p(j|i)}\ketbra{j}{i}, \; i\neq j\right\}$,
as seen in \cref{eq:commKraus}. Furthermore, $\Tr\cwich{\rho\sqrt{p(j|i)}\ketbra{j}{i}}=0$ 
for all $i,j$ because $\rho$ is diagonal and 
\begin{equation}
  \Tr\cwich{\rho \sum_{i}\sqrt{p(i|i)}\ketbra{i}}=\sum_ip_i\sqrt{p(i|i)}.
\end{equation} 
If we input this single nonzero value into the equation we obtain 
\begin{equation}
  \mc{K}_\infty(\rho,\sigma)=1-\sqrt{\swich{\sum_ip_i\sqrt{p(i|i)}}^2}=1-\sum_ip_i\sqrt{p(i|i)},
\end{equation} 
equal to \cref{eq:commStatesLimit}.
\end{remark}

\subsection{Asymmetry and discontinuity of the cost function}
We will consider a similar setting as we had, with the symmetric cost matrix $\tilde{K}_0=\I-\frac{1}{d}\S$, $\rho=\ketbra{0}$ and $\sigma= (1-p)\rho+p\;\I/d$. We now consider the symmetry gap: $\mc{K}(\rho,\sigma)-\mc{K}(\sigma,\rho)$. The cost will be symmetric when this is zero. \cref{fig:SymmetryGap} shows that in the proposed example, this gap is nonzero for all $\sigma\neq\rho.$

 \begin{figure}[ht]
  \centering
    \includegraphics[width=0.6\textwidth]{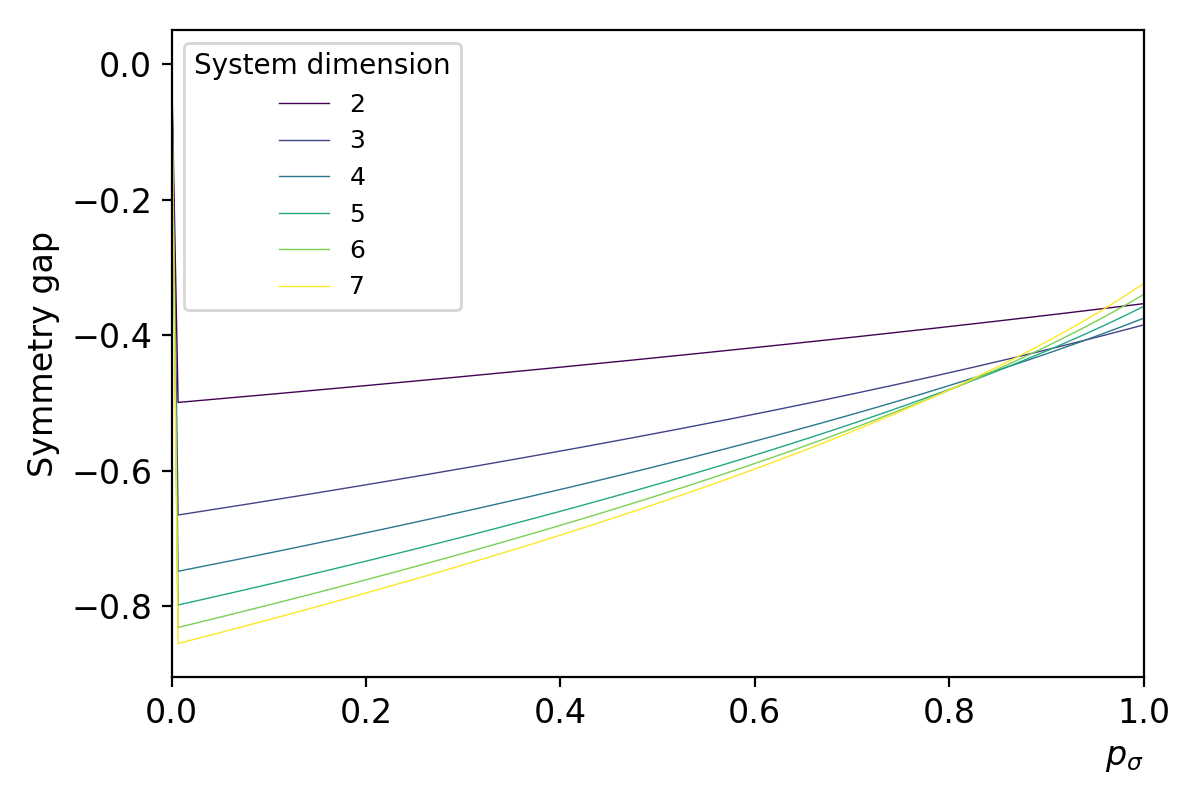} 
    \caption{Symmetry gap between for $\tilde{K}_0=\I-\frac{1}{d}\S$, 
             $\rho=\ketbra{0}$ and $\sigma=(1-p_\sigma)\rho+p_\sigma\I/d$.} 
    \label{fig:SymmetryGap}
\end{figure}

The symmetry gap also shows a discontinuity when $\sigma$ goes to $\rho$. This is due to the following: if a quantum channel takes any non pure state to a pure state, this channel must be the replacement channel due to the continuity of quantum channels. Therefore, in our example,  $\mc{Q}(\sigma,\rho)=\{\sigma\otimes\rho\}$. This is not true anymore for $\sigma=\rho$, since the states over time associated to all the unitary channels that send $\rho$ to itself (including the identity channel) are now feasible. This discontinuity in the feasible set causes a discontinuity in $\mc{K}(\sigma,\rho)$, which translates to the symmetry gap. We can see that analytically with an example. 

Let $\rho=(1-\varepsilon)\ketbra{0}+\varepsilon\I/d$ and $\sigma=\ketbra{+}$. If $\varepsilon>0$, there is a single admissible state over time: $Q=\rho\otimes \sigma$. The associated cost is 
\begin{equation}\begin{split}
    \mc{K}(\rho(\varepsilon),\sigma) &= 1-\frac{1}{d}\Tr\cwich{\S(\rho\otimes\sigma)} 
                                      = 1-\Tr\cwich{\rho\sigma} \\
                                     &= 1-\frac{1}{d}\swich{\frac{1}{2}(1-\varepsilon)+\varepsilon\frac{1}{d}} 
                                      \xrightarrow[\varepsilon \to 0]{} 1-\frac{1}{2d}.
\end{split}\end{equation} 
If we consider $\varepsilon=0$, $\rho$ and $\sigma$ are pure and we can use \cref{prop:analyticalUnitaryMaps} to obtain \begin{equation}
  \mc{K}(\rho(0),\sigma)
     =\frac{1}{d}\swich{1-\frac{1}{\sqrt{2}}}\swich{d+2\frac{1}{\sqrt{2}}}=1-\frac{d+2-2\sqrt{2}}{2d}.
\end{equation} 
As $d+2-2\sqrt{2}>1$ for all $d\geq 2$, we get the strict inequality
\begin{equation}
    \mc{K}(\rho(0),\sigma) < \lim_{\varepsilon \searrow 0}\mc{K}(\rho(\varepsilon),\sigma).
\end{equation} 

\section{Conclusions and open problems}
\label{sec:conclusions}

We have introduced a formulation of optimal transport cost for quantum states as an 
application of the formalism of states over time (\emph{stotes}), in an attempt to base 
it on a notion of cost bilinear in the initial quantum state (mass distribution) 
and quantum channel (transport plan). 
This formalism was introduced to expand on our current understanding of spatial 
quantum correlations, expressed in joint density matrices, to incorporate temporal 
correlations induced by a given time evolution. In it, stotes are Jordan products 
of density matrices with {\jamiol} matrices of quantum channels. 
This has allowed us to define a formalism of optimal transport with a straightforward 
physical interpretation for couplings, albeit outside the realm of density matrices. Despite the clear interpretation of the couplings, and the mathematical properties observed by the stote and state-channel cost function $\kappa$, a physical meaning of the optimal cost remains a crucial open question in this setting.

After introducing the necessary notions, we set out to explore the new definition 
of cost, in particular in view of the possibility of obtaining interesting 
metrics on the set of quantum states. 
The biggest open problem we, and in fact the stote formalism as a whole, face is 
that there is currently no concise characterisation of the convex hull of the set of 
stotes, nor of the convex cone generated by it, nor the dual cone. The latter
cone encodes all information required in the selection of a suitable cost operator:
it should be in the dual cone of stotes, and the same dual cone plays an important
role in deciding the triangle inequality of a given cost. Our original motivation 
was to be able to design cost matrices that can be interpreted in physical terms, 
such as showing energy differences for a given Hamiltonian. Currently this is work 
in progress.
The stote cone itself enters in each attempt of calculating 
the optimal cost for a given cost operator. However, at least fixing cost operator 
as well as initial and final density matrix, this optimisation is an SDP. 


As a case study and because of its distinguished symmetry, we have investigated in detail 
the unitary invariant cost, which is analogous to the trivial metric in the classical 
case. We have calculated this cost for commuting states and pure states. These examples 
have allowed us to observe some properties and facts regarding our formalism. 

A surprising fact can be observed from the second item in \cref{example:UIcost}. One of 
our main motivations for this formalism was the linearity of the state over time with 
respect to both the initial state and the channel. Other approaches to quantum transport 
\cite{depalma21,friedland22} led to the observation that to save the triangle inequality, 
a square root of the cost had to be taken. Likewise, from the example one notices that 
our cost behaves like the square of a distance, indicating that the square root would 
be necessary here, too, to preserve the triangle inequality. 
Nothing similar has been observed in the classical case, where the roots only appear 
when taking powers of distances as cost functions (as seen in the Wasserstein distances \cite{kantorovich60,vaserstein69}). This contrast motivates us to conjecture that the 
appearance of a square root to preserve the triangle inequality could be a quantum 
feature of optimal transport costs.

Other features of our optimal quantum transport cost are that even when the cost 
operator is exchange symmetric (as the unitary invariant is), the resulting 
optimal cost is not necessarily symmetric, adding to doubts that this approach 
can yield meaningful metrics on states. On top of that, the examples of asymmetry 
exhibit even instances of discontinuity of the optimal cost in the two states. 

The second item of \cref{example:UIcost} allows for another observation. 
In every dimension, we considered the initial/final states $\rho=\ketbra{0}$,
$\sigma = (\alpha\ket{0}+ \sqrt{1-\alpha^2}\ket{1})(\alpha\bra{0}+ \sqrt{1-\alpha^2}\bra{1})$. 
The optimal channels turn out to be the conjugation by unitaries having 
a block structure: a direct sum of a $2\times 2$-unitary and an identity of 
rank $d-2$.
Only the first summand is relevant to joint support of $\rho$ and $\sigma$, but the cost 
is a function of the dimension nonetheless, as seen in \cref{eq:UIcost}. The calculation of 
the optimal unitary also shows that the cost for a given channel is sensitive to the behaviour 
of the channel outside the space spanned by the relevant states. 
In fact, this is a general feature for arbitrary states. 
This contrasts classical optimal transport, where the behaviour of the channel on regions 
where the input probability is zero has no effect on the cost. This feature is reminiscent 
of the Aharonov-Bohm effect \cite{aharonov59}, a purely quantum effect where the magnetic field 
far away from a charged particle can affect interference fringes of its wave function. 
Motivated by this, in in \cref{sec:limitCommUI} we considered the limit of larger 
and larger ambient Hilbert spaces, for a given pair of states. This leads to a certain 
renormalisation of the cost (always in the unitary invariant case), in particular 
in the limit we obtain a formula for the optimal transport cost that manifestly 
``feels'' only the supports of the two states. It remains for future investigation 
to determine whether this leads to well-behaved metrics with interesting properties. 



\section*{Acknowledgments}
The first author thanks Marco Fanizza, Giulio Gasbarri, Niklas Galke, Carlo Marconi and Santiago Llorens for helpful discussions.
The authors collectively thank Sheila Folsom for listening to our problems 
about states over time and for showing us how to keep an open mind.
The authors are supported by the Spanish MICIN 
(project PID2022-141283NB-I00) with the support of FEDER funds, 
by the Spanish MICIN with funding from European Union NextGenerationEU 
(project PRTR-C17.I1) and the Generalitat de Catalunya, and by the Spanish MTDFP 
through the QUANTUM ENIA project: Quantum Spain, funded by the European 
Union NextGenerationEU within the framework of the ``Digital Spain 
2026 Agenda''. 
MHR furthermore acknowledges support by the Generalitat de Catalunya through 
a FI-AGAUR scholarship.
JCC furthermore acknowledges support from an ICREA Academia Award. 
AW furthermore acknowledges support by the European Commission QuantERA grant 
ExTRaQT (Spanish MICIN project PCI2022-132965), 
by the Alexander von Humboldt Foundation, 
and by the Institute for Advanced Study of the Technical University Munich. 

The numerical calculations were performed with Python, using the package 
\textit{cvxpy} \cite{diamond16} and the solver \textit{mosek} \cite{mosek}. 
The package \textit{qutip} \cite{johansson12,johansson13} has been used for 
manipulation of quantum objects.

\section*{Author contributions}

All authors contributed to the research and writing of the manuscript. No generative AI has been used in the research or writing of this paper.

\bibliographystyle{quantum}
\bibliography{biblio}

\end{document}